\documentclass{article}

\usepackage[margin=3cm]{geometry}
\usepackage[utf8]{inputenc}

\usepackage{amsmath}
\usepackage{amssymb}
\usepackage{amsthm}
\usepackage{tikz}
\usepackage{subcaption}
\usepackage{enumitem}
\newcommand{\fwf}[1]{10.55776/#1}

\usetikzlibrary{matrix}
\usetikzlibrary{decorations.pathreplacing}

\newtheorem{lemma}{Lemma}
\newtheorem{theorem}{Theorem}

\newtheorem{corollary}{Corollary}
\newtheorem{claim}{Claim}
\theoremstyle{definition}
\newtheorem{definition}{Definition}
\newtheorem{assumption}{Assumption}
\theoremstyle{definition}

\usepackage{thm-restate}
  
\usepackage{xspace}

\usepackage{algorithm}
\usepackage{algorithmic}

\usetikzlibrary{tikzmark}
\usetikzlibrary{decorations}
\usepackage{makecell}

\makeatletter
\renewcommand*\env@matrix[1][\arraystretch]{%
  \edef\arraystretch{#1}%
  \hskip -\arraycolsep
  \let\@ifnextchar\new@ifnextchar
  \array{*\c@MaxMatrixCols c}}
\makeatother

\newcommand{\calF}{\mathcal{F}}

\newcommand{\calD}{\mathcal{D}}
\newcommand{\calP}{\mathcal{P}}
\newcommand{\calT}{\mathcal{T}}
\newcommand{\calC}{\mathcal{C}}
\newcommand{\var}{\mathit{var}}
\newcommand{\clauses}{\mathit{clauses}}

\newcommand{\tw}{\mathit{tw}}
\newcommand{\tseitin}{\textit{Tseitin}}
\newcommand{\gates}{\mathit{gates}}

\newcommand{\elements}{\mathit{set}}
\newcommand{\poly}{\mathit{poly}}

\newcommand*\circled[2]{\tikz[baseline=(char.base)]{
            \node[shape=circle,draw=#2,thick,inner sep=1pt,font=\small] (char) {#1};}}

\title{\textbf{Compilation and Fast Model Counting beyond CNF}}

\author{
\textbf{Alexis de Colnet}, \textbf{Stefan Szeider} and \textbf{Tianwei Zhang}
\bigskip\\
{Algorithms and Complexity Group, TU Wien, Vienna, Austria}
\\
\{decolnet,sz,zhangtw\}@ac.tuwien.ac.at
}
          \date{ }  
            
\begin{document}

\maketitle

\begin{abstract}
Circuits in deterministic decomposable negation normal form (d-DNNF)  are representations of Boolean functions that enable
linear-time model counting. This paper strengthens our theoretical
knowledge of what classes of functions can be efficiently transformed,
or \emph{compiled}, into d-DNNF. Our main contribution is the
fixed-parameter tractable (FPT) compilation of conjunctions of specific
constraints parameterized by incidence treewidth. This
subsumes the known result for CNF. The constraints in question are all
functions representable by constant-width ordered binary decision
diagrams (OBDDs) for \emph{all} variable orderings. For instance, this
includes parity constraints and cardinality constraints with constant
threshold. The running time of the FPT compilation is singly
exponential in the incidence treewidth but hides large constants in the exponent. To balance that, we give a more efficient FPT algorithm for model counting that applies to a sub-family of the constraints and does not require compilation.
\end{abstract}

\section{Introduction}

Knowledge compilation is a domain of computer sciences that studies
the different ways to represent functions. Classes of representations,
or \emph{languages}, have been invented where specific problems become
tractable. In particular, many Boolean languages have been created
that support polynomial-time model counting, that is, determining the
number of truth assignments on which a function $f : \{0,1\}^n
\rightarrow \{0,1\}$ evaluates to~$1$. In practice, several model
counters for CNF formulas are transforming, or \emph{compiling}, their inputs into the language d-DNNF of circuits in deterministic decomposable negation normal form, where model counting is feasible in linear time (in the size of the circuits). For instance, \texttt{C2D}~\cite{Darwiche04}, \texttt{DSharp}~\cite{MuiseMBH12}, \texttt{miniC2D}~\cite{OztokD15}, and \texttt{D4}~\cite{LagniezM17} all compile CNF formulas into (sublanguages) of d-DNNF. Other model counters do not explicitly compile the CNF but can be seen as compilers in disguise~\cite{KieselE23}. 

Compilation to d-DNNF is often hard in the sense that the d-DNNF
circuit representations of many functions provably have exponential
size, even for functions that belong to fragments where model counting
is tractable such as systems of Boolean linear
equations~\cite{deColnetM23}. On a more positive note, compiling from
CNF to d-DNNF is fixed-parameter tractable (FPT) when the parameter is
the treewidth of the CNF's primal graph~\cite{HuangD07} and when it is
the treewidth of its incidence graph~\cite{BovaCMS15}. To further
understand when the input is easy to compile, one possibility is to
find new CNF parameters that dominate treewidth and yet enable FPT
compilation to d-DNNF, for instance, decision-width~\cite{OztokD14},
CV-width~\cite{OztokD14a} and PS-width~\cite{BovaCMS15}. Another
direction is to part from the CNF input and study the parameterized
compilability of functions given in a different format such as DNF or
a general Boolean circuit~\cite{AmarilliCMS20}. In this paper, we
follow the latter direction: we investigate the FPT compilation of
systems (i.e., conjunctions) of Boolean constraints that are not just
CNF. We see each CNF formula as a system of disjunctive clauses and
identify other types of constraints that can be added to the system
while ensuring an FPT compilation to d-DNNF parameterized by its
incidence treewidth $k$. While compiling general systems of
constraints to d-DNNF cannot be FPT parameterized $k$ unless
$\text{FPT} = \text{W}[1]$ (since even satisfiability is
$\text{W}[1]$-hard in this setting~\cite{SamerS10}), we show that
there is an FPT compilation algorithm for specific constraints.  In a
first approximation, these are the constraints that can only be in a
constant number of states when assigning any subset of the variables
in any way. Parity constraints (i.e., constraints of the form ``the
sum of the variables is odd/even'') are a good example where there are
only two possible states: given a parity constraint and a subset of
its variables, assigning these variables in any way either results in
the odd parity constraint or the even parity constraint on the
remaining variables. Our main result is the following.

\begin{theorem}\label{theorem:intro_thm}
Let $w$ be a constant integer and $\calC$ be a class of constraints such that, for every $c \in \calC$, and every subset $Y$ of $c$'s variables, at most $w$ different constraints can be obtained from $c$ by assigning variables in $Y$ in any way. There is an algorithm that, given a system $F$ of constraints from $\calC$, returns in time $2^{O(k)}\poly(|F|+|var(F)|))$ a d-DNNF circuit for~$F$, where $k$ is the treewidth of the incidence graph of $F$.
\end{theorem}

\noindent The neatest proof we could find for this result on the compilation of non-CNF systems---which we present in this paper---goes through constructing CNF encodings of the systems and then compiling these encodings to d-DNNF.

Theorem~\ref{theorem:intro_thm} includes a rather lengthy
characterization of our constraints. For the proof, we use a ``more
advanced'' characterization in knowledge compilation terms. The class
$\calC$ will be a class of constraints that are \emph{$w$-slim} for
the language OBDD or the language SDNNF. This notion is inspired by
Wegener's book~\cite[Section 5.3]{Wegener00} and introduced in
detail in Section~\ref{section:slim_constraint_STS}. We are not the first to consider OBDD constraints to prove FPT results, for instance~\cite{ChenG10a} use similar constraints to prove FPT results for the constraint satisfaction problem. We prove our theorem
in Section~\ref{section:fpt_compilation} with a more accurate running
time where the exponential part also depends
on~$w$. Finally, in Section \ref{section:model_counting}, we show that we can
count models faster than by compilation if we put additional
restrictions on the constraints.  One of our results is
based on bounded state-size (introduced in
Section~\ref{section:slim_constraint_STS}), which captures important
constraints like XORs, modulo, and cardinality constraints.

\begin{theorem}
Let $F$ be a system of constraints whose maximum state-size is $w$. There is an algorithm that, given $F$ and a nice width-$k$ tree decomposition of the incidence graph of~$F$, counts the models of $F$ in time $O(w^{2k}(|F|+|var(F)|))$ in the unit-cost model. 
\end{theorem}

\section{Preliminaries}\label{section:A preliminaries $x$}

A Boolean variable $x$ takes a value in $\{0,1\}$. \emph{Literals} are
denoted by $x$ and $\bar x$. An \emph{assignment} to a set $X$ of
Boolean variables is a mapping from $X$ to $\{0,1\}$, and a
\emph{Boolean function} $f$ over $X$ maps the assignments to $X$ to
$\{0,1\}$. For us, \emph{constraints} are just Boolean functions that
appear in systems (conjunctions of constraints). \emph{CNF formulas}
are systems of constraints whose constraints are \emph{clauses} (disjunctions
of literals). We write $|F|$ for the number of constraints in
the system~$F$. Boolean circuits are another way to represent functions. The size
of a circuit~$D$, denoted $|D|$, is its number of gates and
connectors. We denote by $\gates(D)$ the set of gates of $D$. For $g
\in \gates(D)$, $D_g$ is the subcircuit of $D$ whose output gate is
$g$. We write $f(X)$, $D(X)$, and $F(X)$ to indicate that the variable
set of a function, circuit, or system is $X$, respectively. If this set is not explicit, we use $\var(f)$, $\var(D)$, $\var(F)$.

\subsection{Treewidth and Tree Decompositions}

A \emph{tree decomposition} $\calT$ of a graph $G$ is a pair $(T,b)$ with $T$ a tree and $b : V(T) \rightarrow \calP(V(G))$ a \emph{bag} function such that
(1)~$\bigcup_{t\in V(T)} b(t) = V(G)$,
(2)~for all $uv \in E(G)$, there is $t\in T$ such that $\{u,v\}
  \subseteq b(t)$, and
(3)~for all  $v \in V(G), T[t \mid v \in b(t)]$ is connected.
The \emph{width} of $\calT$ is $\max_{t \in T}|b(t)|$. The
\emph{treewidth} of $G$ is $\tw(G) := \min_\calT \max_{t \in T}|b(t)|
- 1$ where $\calT$ ranges over all tree decompositions of $G$. A tree
decomposition is \emph{nice} when it is rooted and when each node $t
\in V(T)$ is of one of the following three types: a join node, an
introduce node, or a forget node. $t$ is a \emph{join node} if it has two children $t_1$ and $t_2$ and $b(t) = b(t_1) = b(t_2)$. $t$ is an \emph{introduce node} for $v \in V(G)$ if it has a single child $t'$ and $v \not\in b(t')$ and $b(t) = b(t') \cup \{v\}$. $t$ is a \emph{forget node} for $v \in V(G)$ if it has a single child $t'$ and $v \not\in b(t)$ and $b(t') = b(t) \cup \{v\}$. In addition, the bag of the root node is empty. Every tree decomposition can be made nice without increasing its width.

The \emph{incidence graph} of a system $F$ of constraints, denoted by $G_F$, is the graph whose vertices are the constraints and the variables of $F$ and such that there is an edge between a constraint $c$ and a variable $x$ if and only if $x \in \var(c)$. The \emph{incidence treewidth} of $F$ is $\tw_i(F) := \tw(G_F)$.

\subsection{OBDDs and SDNNF Circuits}

\paragraph{(d-)SDNNF Circuits.} A variable tree (vtree) $\tau$ over a
set $X$ of variables is a rooted binary tree whose leaves are in bijection with $X$. For every $t \in V(\tau)$, $\var(t)$ is the set of variables on the leaves below $t$. A circuit $D$ in structured-decomposable negation normal form (SDNNF)~\cite{PipatsrisawatD08} is a Boolean circuit with literal inputs, whose gates are binary $\lor$-gates and binary $\land$-gates, and such that there exist a vtree $\tau$ over $\var(D)$ and a mapping $\lambda : \gates(D) \rightarrow V(\tau)$ verifying the following:
\begin{itemize}
\item[•] For every $g \in \gates(D)$, $\var(D_g) \subseteq \var(\lambda(g))$.
\item[•] For every $\lor$-gate $g = g_1 \lor g_2$, $\lambda(g) = \lambda(g_1) = \lambda(g_2)$.
\item[•] For every $\land$-gate $g = g_1 \land g_2$, $\lambda(g)$ has two children $t_1$ and $t_2$ and $\lambda(g_1)$ is below $t_1$ and $\lambda(g_2)$ is below $t_2$.
\end{itemize}
\noindent We say that $D$ is \emph{structured by} $(\tau,\lambda)$ and
sometimes omit $\lambda$. An example of an SDNNF circuit is shown in Figure~\ref{figure:SDNNF} with the mapping $\lambda$ given by the color code. An SDNNF circuit $D$ is deterministic (d-SDNNF) when for all $\lor$-gates $g = g_1 \lor g_2$, we have $D^{-1}_{g_1}(1) \cap D^{-1}_{g_2} = \emptyset$. Counting the models of $D$, i.e., finding $|D^{-1}(1)|$, is tractable when $D$ is a d-SDNNF circuit. 

\paragraph{OBDDs.} A binary decision diagram (BDD) is a directed
acyclic graph with a single source, two sinks labeled $0$ and~$1$,
whose internal nodes are decision nodes with two distinct children and
labeled by variables. A node $v$ labeled by variable $x$ and with
children $v_0$ and $v_1$ is recursively interpreted as a Boolean
function $v = (\bar x \land v_0) \lor (x \land v_1)$. Every assignment
$\alpha$ corresponds to a path in the DAG. Starting from the root,
assuming $\alpha$'s path reaches $v$ follows $v_0$ if $\alpha(x) = 0$
and $v_1$ otherwise. The assignments satisfying the BDD are exactly
those whose paths reach the sink $1$. For a total order $\pi$ on the
variables, a $\pi$-Ordered BDD ($\pi$-OBDD) is a BDD whose variables
appear at most once along every path from the root to a sink and
always in an order consistent with~$\pi$. A $\pi$-OBDD is
\emph{complete} if every path contains all the variables. d-SDNNF circuits generalize OBDDs in the sense that there is a simple linear-time rewriting that transforms  OBDDs into d-SDNNF circuits structured by \emph{linear vtrees} (i.e., vtrees whose internal nodes all have a leaf child).

\paragraph{Width measures.} The \emph{width} of an OBDD is the maximum number of nodes labeled by the same variable. Note that making an OBDD complete, while feasible in linear time, can also increase the width by a linear factor~\cite{BolligW00}. The \emph{width} of an SDNNF circuit structured by $(T,\lambda)$ is defined as $\max_{t \in V(T)} |\lambda^{-1}(t)|$. This definition differs from that of~\cite{CapelliM19} but is more convenient for stating our results.

\section{Slim Functions and STS}\label{section:slim_constraint_STS}

In this paper, constraints are \emph{global constraints}~\cite{HoeveK06}. So we know the type (e.g., clauses, parity constraints, cardinality constraints) of every constraint we manipulate.

CNF formulas are known to be FPT compilable to d-SDNNF parameterized
by their incidence treewidth.

\begin{theorem}[\cite{BovaCMS15}]\label{theorem:cnf_fpt_compile}
There is an algorithm that transforms any CNF formula $F$ into an equivalent d-SDNNF circuit in time $2^{O(\tw_i(F))}\poly(|F| + |\var(F)|)$.
\end{theorem}

\noindent The result leans on the fact that for every clause $c$,
every subset $Y \subseteq \var(c)$, and every assignment
$\alpha$ to $Y$, the clause $c|\alpha$ where variables are assigned as
in $\alpha$ can be in only two states: either it is satisfied, or it
is the projection of $c$ onto the complement of $Y$. Note that
$\alpha$ falsifying $c$ falls into the second case ($Y = \var(c)$ and
$c|\alpha$ is the empty clause). Our intuition is that FPT compilation
to d-SDNNF should be possible for systems of constraints with a
similar property: when subject to any assignment to any
fixed set $Y$, every constraint of the system can only be in a few states. For example, consider an XOR constraint $c : x_1 \oplus x_2 \oplus x_3 \oplus x_4 \oplus x_5 = 0$ and $Y = \{x_1,x_2,x_3\}$. Then $c|\alpha$ is either the odd parity constraint $x_4 \oplus x_5 = 1$, or the even parity constraint $x_4 \oplus x_5 = 0$.

\subsection{Slim Functions for Complete OBDDs and SDNNFs}\label{section:slim_constraint}

One can consider  a \emph{complete} OBDD is as state diagram where the values of the variables are read in a predefined order. The width of the complete OBDD is then the maximum number of states reachable after having set any number of variables in that order. The property we want for our constraints translates into the requirement that the smallest possible width in one of their $\pi$-OBDD representations is bounded for \emph{all} the variable orders~$\pi$. In the following, $h : \mathbb{N} \rightarrow \mathbb{R}$ is a real function.
 
\begin{definition}
A class $\calF$ of Boolean functions is \emph{$h$-slim for complete OBDDs} when, for every $n$-variables function $f \in \calF$ and every total order $\pi$ of $\var(f)$, there is a complete $\pi$-OBDD of width at most $h(n)$ computing $f$.
\end{definition} 

The notions of $O(1)$-slim functions, $O(n)$-slim functions, etc.\
should be self-explanatory. Given a fixed constant $w$, we talk of
$w$-slim functions when $h$ is the constant~$w$ function. When $h$ is
an unknown polynomial, the $h$-slim functions for complete OBDDs
coincide with the \emph{nice functions} from Wegener's
book~\cite[Section 5.3]{Wegener00}. Our proofs below rely on SDNNF
representations rather than OBDD representations. The concept of
$h$-slim functions generalizes to SDNNF.

\begin{definition}
A class $\calF$ of Boolean functions is \emph{$h$-slim for complete SDNNFs} when, for every $n$-variables function $f \in \calF$ and vtree $T$ over $\var(f)$, there is a complete SDNNF of width at most $h(n)$ and with vtree $T$ that computes~$f$.
\end{definition} 

Given the rewriting of complete OBDDs into complete d-SDNNFs, it is immediate that having a small OBDD-width for every variable ordering implies having a small
SDNNF-width for every \emph{linear} vtree. However, generalizing to
\emph{all} possible vtrees is not straightforward.

\begin{lemma}\label{lemma:slimOBDDgiveNiceSDNNF}
If a class $\calF$ of function is $h$-slim for complete OBDDs, then it is $O(h^3)$-slim for complete SDNNFs. 
\end{lemma}
\begin{proof}
Let $f \in \calF$. Let $\tau$ be a vtree over $X = \{x_1,\dots,x_n\}$ and suppose the left-to-right ordering of $\tau$'s leaves is the natural ordering $\pi : x_1 \prec \dots \prec x_n$. We show how to construct an SDNNF circuit respecting $\tau$ and computing $f$. For every $1 \leq i < j \leq n$ we define the \emph{$[i,j]$-slice} of $X$, denoted by $X_{[i,j]}$, as $\{x_i,x_{i+1},\dots,x_j\}$. By definition of $\pi$, for every $t \in V(\tau) \setminus \mathit{leaves}(\tau)$, the vtree $\tau_t$ rooted under $t$ in $\tau$ is a vtree over a slice $X_{[i_t,j_t]}$. We define $S \subseteq V(\tau)$ as the set of nodes $t$ such that $\tau_t$ is right-linear and such that, if $t$ has a parent $t'$, then $\tau_{t'}$ is not linear. The set of slices $X_{[i_t,j_t]}$ for all $t \in S$ forms a partition of $X$. 

The construction uses a $\pi$-OBDD $B$ representing $f$. The levels of $B$ are numbered top-down: level $1$ corresponds to the decision nodes for $x_1$ (so only the root of $B$), $\dots$, level $n$ corresponds to the decision nodes for $x_n$, and level $n+1$ comprises the two sinks of $B$. Let $w_i$ be the number of nodes at level $i$ of $B$ and introduce a symbol for each of these nodes: $s^i_1,\dots,s^i_{w_i}$. Note $w := \max_{i \in [n]}w_i$. We also define $w_{n+1} = 1$ and keep the symbols $0$ for $s^{n+1}_1$ and $1$ for $s^{n+1}_2$. For every $1 \leq i < j \leq n$, we define the \emph{$[i,j]$-slice} of $B$, denoted by $B_{[i,j]}$, as the function from $\{0,1\}^{X_{[i,j]}}$ to $\{s^{j+1}_1,\dots,s^{j+1}_{w_{j+1}}\}^{w_i}$ obtained by taking all levels of $B$ from the $i$th to the $(j+1)$th level and by replacing every node at the $(j+1)$th level by the corresponding symbol $s^{j+1}_{\ast}$. So $B_{[i,j]}$ is an OBDD with $w_i$ roots corresponding to $s^{i}_1,\dots,s^{i}_{w_i}$ and whose sinks are labelled by the symbols $s^{j+1}_1,\dots,s^{j+1}_{w_{j+1}}$. To be clear: the sinks of $B_{[i,j]}$ are \emph{explicitly} labelled by symbols $s^{j+1}_{\ast}$, and the roots are \emph{implicitly} mapped to the symbols $s^i_\ast$. An example of such an OBDD is shown Figure~\ref{figure:obdd_slice}.

\begin{figure}[h!]
\centering
\begin{subfigure}{0.4\textwidth}
\centering
\begin{tikzpicture}
\def\s{1.5};
\node[circle,draw,inner sep=\s,label={north:$s^3_1$}] (x11) at (0,0) {$x_3$};
\node[circle,draw,inner sep=\s,label={north:$s^3_2$}] (x12) at (1,0) {$x_3$};
\node[circle,draw,inner sep=\s,label={north:$s^3_3$}] (x13) at (2,0) {$x_3$};

\node[circle,draw,inner sep=\s] (x21) at (0,-1) {$x_4$};
\node[circle,draw,inner sep=\s] (x22) at (1,-1) {$x_4$};
\node[circle,draw,inner sep=\s] (x23) at (2,-1) {$x_4$};

\node[circle,draw,inner sep=\s] (x31) at (0.5,-2) {$x_5$};
\node[circle,draw,inner sep=\s] (x32) at (1.5,-2) {$x_5$};

\node[rectangle,draw,inner sep=\s] (s1) at (0,-3) {$s^6_1$};
\node[rectangle,draw,inner sep=\s] (s2) at (1,-3) {$s^6_2$};
\node[rectangle,draw,inner sep=\s] (s3) at (2,-3) {$s^6_3$};

\draw[-latex,dashed] (x11) -- (x21);
\draw[-latex] (x11) -- (x22);
\draw[-latex,dashed] (x12) -- (x22);
\draw[-latex] (x12) -- (x23);
\draw[-latex,dashed] (x13) -- (x23);
\draw[-latex] (x13) -- (x22);

\draw[-latex,dashed] (x21) -- (x31);
\draw[-latex] (x21) -- (x32);
\draw[-latex,dashed] (x22) -- (x32);
\draw[-latex] (x22) -- (x31);
\draw[-latex,dashed] (x23) to [out=-90,in=40] (x32);
\draw[-latex] (x23) to [out=-140,in=80] (x32);

\draw[-latex,dashed] (x31) -- (s1);
\draw[-latex] (x31) -- (s2);
\draw[-latex,dashed] (x32) -- (s2);
\draw[-latex] (x32) -- (s3);
\end{tikzpicture}
\caption{A ``slice'' of an OBDD}\label{figure:obdd_slice}
\end{subfigure}\begin{subfigure}{0.6\textwidth}
\centering
\hfill\begin{tikzpicture}
\def\s{1.5};
\node[circle,draw,inner sep=\s,label={north:$s^3_1$}] (x11) at (0,0) {$x_3$};

\node[circle,draw,inner sep=\s] (x21) at (0,-1) {$x_4$};
\node[circle,draw,inner sep=\s] (x22) at (1,-1) {$x_4$};

\node[circle,draw,inner sep=\s] (x31) at (0.5,-2) {$x_5$};
\node[rectangle,draw,inner sep=\s] (s0) at (1.5,-2) {$0$};

\node[rectangle,draw,inner sep=\s] (s1) at (0,-3) {$1$};

\draw[-latex,dashed] (x11) -- (x21);
\draw[-latex] (x11) -- (x22);

\draw[-latex,dashed] (x21) -- (x31);
\draw[-latex] (x21) -- (s0);
\draw[-latex,dashed] (x22) -- (s0);
\draw[-latex] (x22) -- (x31);

\draw[-latex,dashed] (x31) -- (s1);

\node[rectangle,draw,inner sep=\s] (s0) at (1,-3) {$0$};
\draw[-latex] (x31) -- (s0);
\def\o{2.5};

\node[circle,draw,inner sep=\s,label={north:$s^3_1$}] (x11) at (\o+0,0) {$x_3$};

\node[circle,draw,inner sep=\s] (x21) at (\o+0,-1) {$x_4$};
\node[circle,draw,inner sep=\s] (x22) at (\o+1,-1) {$x_4$};

\node[circle,draw,inner sep=\s] (x31) at (\o+0.5,-2) {$x_5$};
\node[circle,draw,inner sep=\s] (x32) at (\o+1.5,-2) {$x_5$};

\node[rectangle,draw,inner sep=\s] (s1) at (\o+0,-3) {$0$};
\node[rectangle,draw,inner sep=\s] (s2) at (\o+1,-3) {$1$};
\node[rectangle,draw,inner sep=\s] (s3) at (\o+2,-3) {$0$};

\draw[-latex,dashed] (x11) -- (x21);
\draw[-latex] (x11) -- (x22);

\draw[-latex,dashed] (x21) -- (x31);
\draw[-latex] (x21) -- (x32);
\draw[-latex,dashed] (x22) -- (x32);
\draw[-latex] (x22) -- (x31);

\draw[-latex,dashed] (x31) -- (s1);
\draw[-latex] (x31) -- (s2);
\draw[-latex,dashed] (x32) -- (s2);
\draw[-latex] (x32) -- (s3);

\def\o{5};

\node[circle,draw,inner sep=\s,label={north:$s^3_1$}] (x11) at (\o+0,0) {$x_3$};

\node[circle,draw,inner sep=\s] (x21) at (\o+0,-1) {$x_4$};
\node[circle,draw,inner sep=\s] (x22) at (\o+1,-1) {$x_4$};

\node[rectangle,draw,inner sep=\s] (x31) at (\o+0.5,-2) {$0$};
\node[circle,draw,inner sep=\s] (x32) at (\o+1.5,-2) {$x_5$};

\node[rectangle,draw,inner sep=\s] (s2) at (\o+1,-3) {$0$};
\node[rectangle,draw,inner sep=\s] (s3) at (\o+2,-3) {$1$};

\draw[-latex,dashed] (x11) -- (x21);
\draw[-latex] (x11) -- (x22);

\draw[-latex,dashed] (x21) -- (x31);
\draw[-latex] (x21) -- (x32);
\draw[-latex,dashed] (x22) -- (x32);
\draw[-latex] (x22) -- (x31);

\draw[-latex,dashed] (x32) -- (s2);
\draw[-latex] (x32) -- (s3);

\end{tikzpicture}\hfill\text{}
\caption{The OBDD $B_{[3,5]}$, $B_{s^{3}_{1} \rightarrow s^{6}_{1}}$, $B_{s^{3}_{1} \rightarrow s^{6}_{2}}$ and $B_{s^{3}_{1} \rightarrow s^{6}_{3}}$}\label{figure:split_OBDD}
\end{subfigure}
\caption{}
\end{figure}

From $B_{[i,j]}$ we create a collection $\calD_{[i,j]} = \{ D_{s^i_{k} \rightarrow s^{j+1}_{\ell}} \mid k \in [w_i], \ell \in [w_{j+1}]\}$ of complete SDNNF circuits over $X_{[i,j]}$ all structured by $\tau$ and such that $D_{s^i_{k} \rightarrow s^{j+1}_{\ell}}$ accepts exactly the assignments to $X_{[i,j]}$ for which the OBDD rooted under the $k$th root of $B_{[i,j]}$ evaluates to $s^{j+1}_{\ell}$.

\begin{claim}\label{claim:split_OBDD}
For every $t \in S$, from $B_{[i_t,j_t]}$ we can construct $\calD_{[i_t,j_t]}$ such that every SDNNF circuit in $\calD_{[i_t,j_t]}$ respects $\tau_t$ and has width $O(w)$.
\end{claim}
\begin{proof}
To construct $D_{s^{i_t}_{k} \rightarrow s^{j_t+1}_{\ell}}$, we keep only the OBDD rooted under the $k$th root of $B_{[i_t,j_t]}$, next we replace by the $0$-sink all sinks different from $s^{j_t+1}_{\ell}$, next we replace the $s^{j_t+1}_{\ell}$-sink by a $1$-sink. This gives a coomplete $\pi$-OBDD $B_{s^{i_t}_{k} \rightarrow s^{j_t+1}_{\ell}}$ that accepts exactly the assignments to $X_{[i_t,j_t]}$ for which the OBDD under the $k$th root of $B_{[i_t,j_t]}$ evaluates to $s^{j_t+1}_{\ell}$. See Figure~\ref{figure:split_OBDD} for an example. Finally we use the classical translation from complete OBDD to complete SDNNF to obtain $D_{s^{i_t}_{k} \rightarrow s^{j_t+1}_{\ell}}$. We have that $\mathit{width}(D_{s^{i_t}_{k} \rightarrow s^{j_t+1}_{\ell}}) = O(\mathit{width}(B_{s^{i_t}_{k} \rightarrow s^{j_t+1}_{\ell}})) = O(\mathit{width}(B_{[i_t,j_t]})) = O(w)$.
\end{proof}

Call $\mathit{ancestors}(S)$ the set of strict ancestor nodes of $S$ in $\tau$. The construction of $D$ follows a process that visits in $S \cup \mathit{ancestors}(S)$ in depth-first order.
\begin{itemize}
\item[•] \textbf{Base case.} For $t \in S$, we construct $B_{[i_t,j_t]}$ and the corresponding $\calD_{[i_t,j_t]}$ as in Claim~\ref{claim:split_OBDD}.
\item[•] \textbf{Inductive case.} For $t \in \mathit{ancestors}(S)$ with children $a$ and $b$ such that we have already constructed $\calD_{[i_a,j_a]}$ and $\calD_{[i_b,j_b]}$, by definition of $\pi$ we have that $i_t = i_a$, $i_b = j_a + 1$, and $j_t = j_b$. For every $k \in [w_{i_t}]$ and $\ell \in [w_{j_t+1}]$, we construct $D_{s^{i_t}_{k} \rightarrow s^{j_t+1}_{\ell}}$ as
\begin{equation}\label{eq:induction_obdd_to_sdnnf}
D_{s^{i_t}_{k} \rightarrow s^{j_t+1}_{\ell}} = \bigvee_{h \in [w_{j_a+1}]} (D_{s^{i_t}_{k} \rightarrow s^{j_a+1}_{h}}) \land (D_{s^{i_b}_{h} \rightarrow s^{j_t+1}_{\ell}})
\end{equation}
By induction hypothesis, all $D_{s^{i_t}_{k} \rightarrow s^{j_a+1}_{h}}$ and $D_{s^{i_b}_{h} \rightarrow s^{j_t+1}_{\ell}}$ are complete SDNNF circuits structured by $\tau_a$ and $\tau_b$, respectively, so each new $\land$-gate is decomposable and respects the vtree $\tau_t$. In addition the $\lor$-gates created with the $\bigvee_{h \in [w_{j_a + 1}]}$ have the same variable sets because, by completeness assumption, $var(D_{s^{i_t}_{k} \rightarrow s^{j_a+1}_{h}}) = var(D_{s^{i_t}_{k} \rightarrow s^{j_a+1}_{h'}})$ and $var(D_{s^{i_b}_{h} \rightarrow s^{j_t+1}_{\ell}}) = var(D_{s^{i_b}_{h'} \rightarrow s^{j_t+1}_{\ell}})$for all $h \neq h'$.  Thus $D_{s^{i_t}_{k} \rightarrow s^{j_t+1}_{\ell}}$ is a complete SDNNF respecting $\tau_t$. Given the functions computed by the $D_{s^{i_t}_{k} \rightarrow s^{j_a+1}_{h}}$-s and the $D_{s^{i_b}_{h} \rightarrow s^{j_t+1}_{\ell}}$-s, we have that $D_{s^{i_t}_{k} \rightarrow s^{j_t+1}_{\ell}}$ accepts exactly the assignments to $X_{[i_t,j_t]}$ on which the OBDD rooted under the $k$th root of $B_{[i_t,j_t]}$ evaluates to $s^{j_t+1}_{\ell}$.
\end{itemize}

Th construction ends the root $r$ of $\tau$ is visited. Since $w_{[r]} = 1$, $i_r = 1$, $j_r = n$ and $X_{[i_r,j_r]} = X$, we end up with the set $\calD_{[1,n]} = \{D_{s^{1}_{k} \rightarrow s^{n+1}_{1}}, D_{s^{1}_{1} \rightarrow s^{n+1}_{2}}\} = \{D_{s^{1}_{1} \rightarrow 0}, D_{s^{1}_{1} \rightarrow 1}\}$. The inductive argument gives that $D_{s^{1}_{1} \rightarrow 1}$ is a complete SDNNF circuit respecting $T$ and that accepts exactly the assignments to $X$ on which the OBDD $B$ evaluates to $1$. So $D := D_{s^{1}_{1} \rightarrow 1}$. Finally we look at the width of $D$. The number of gates in $D$ mapped to a given node below $T_t$ for some $t \in S$ is at most $\sum_{D \in \calD_{[i_t,j_t]}} \mathit{width}(D) \leq O(w)|\calD_{[i_t,j_t]}| \leq O(w^3)$. For a node $t \in \mathit{ancestors}(S)$, the gates mapped to it are exactly the $\lor$-gates and the $\land$-gates created in~(\ref{eq:induction_obdd_to_sdnnf}) for all possible $k$ and $\ell$. For a fixed $k$ and $\ell$, $D_{s^{i_t}_{k} \rightarrow s^{j_t+1}_{\ell}}$ contributes $O(w)$ gates, so for all possible $k$ and $\ell$ we end up with $O(w^3)$ gates of $D$ mapped to $t$. 
\end{proof}

\begin{lemma}
If a class $\calF$ of functions is $h$-slim for complete SDNNFs, then it is $2^{2^{O(h)}}$-slim for complete OBDDs.
\end{lemma}
\begin{proof}[Proof sketch]
A \emph{right-linear} vtree over a set $X$ of variables is a vtree whose internal nodes all have a leaf (so a variable) for left-children. Orderings of $X$ are in bijection with the right-linear vtrees over $X$. We just have to show that every SDNNF circuit $D$ structured by a right-linear vtree corresponding to the variable ordering $\pi$ can be transformed into a $\pi$-OBDD of width $2^{2^{O(width(D))}}$. Let $w := width(D)$. The proof is the combinations of  three results: 
\begin{itemize}
\item[(1)] it is possible to encode $D(X)$ in a CNF $F(X,Y)$ with the Tseitin encoding such that the \emph{incidence pathwidth} of $F$ is $p = O(w)$;
\item[(2)] it is possible to compile the $F(X,Y)$ into an $\sigma$-OBDD $B(X,Y)$ of width $2^p$, where $\sigma$ is a variable ordering of $X\cup Y$ whose restriction to $X$ is $\pi$;  
\item[(3)] using~\cite[Lemma 1]{CapelliM19}, we construct an OBDD $B'(X)$ of width at most $2^{width(B)}$ that computes $\exists Y.B(X,Y) \equiv \exists Y. F(X,Y) \equiv D(X)$ and whose variable ordering on $X$ is consistent with $\sigma$ (so $B'$ is a $\pi$-OBDD). 
\end{itemize} 
\end{proof}

\noindent The doubly exponential upper bound could be a bit loose but cannot be decreased below exponential due to the following.

\begin{lemma}\label{lemma:slimSDNNF_do_not_give_slimOBDD}
There are functions computed by complete SDNNFs of width $O(n)$ for every vtree, but that are only computed by OBDDs and d-SDNNFs of width $2^{\Omega(n)}/n$.
\end{lemma}
\begin{proof}
Every conjunction of literals can be turned into an SDNNF circuit of width $O(1)$ for any vtree, thus every DNF formulas $F$ can be turned into an SDNNF circuit of width $O(|F|)$. There exist CNF formulas over $n$ variables that comprise $O(n)$ clauses and whose DNNF representations all contain $2^{\Omega(n)}$ gates~\cite{BovaCMS14}. Every OBDD can be transformed in DNNF in linear time and OBDDs can be negated in constant time (inverting the two sinks). Thus $F$ can only be represented by OBDD of size $2^{\Omega(n)}$ and thus of width $2^{\Omega(n)}/n$.
\end{proof}

Let us now restate our main result, i.e., Theorem~\ref{theorem:intro_thm}, using constraints that belong to families of functions $O(1)$-slim for complete SDNNFs. Recall that our constraints are global. We will need the following assumption. 

\begin{assumption}\label{assumption}
  For every type of constraint appearing in our system, we have a
  polynomial-time algorithm to compile every constraint of this type
  into a minimal-width complete SDNNF for any given vtree.
\end{assumption} 

\begin{theorem}\label{theorem:main_fpt_compile}
Let $w \in \mathbb{N}$ be a fixed constant. Under assumption~\ref{assumption}, there is an algorithm that, given a system $F$ of constraints that are all $w$-slim for complete SDNNFs, constructs in time $2^{O(w\cdot \tw_i(F))}poly(|F|+|\var(F)|+w)$ a circuit in d-SDNNF that computes $F$.
\end{theorem}

In Table~\ref{table:constraints}, we give several families of constraints along with upper bounds on the smallest $h$ functions for which they are $h$-slim for complete OBDDs and SDNNFs. The proof of the correctness of these values for $h$ appear in the long version of the paper. Assumption~\ref{assumption} is reasonable for most constraint types of the table~\ref{table:constraints}, in particular for clauses, XORs, sums modulo, cardinality and threshold constraints.

\subsection{Commutative State Transitions Systems}\label{section:STS}

\begin{definition}
    A \emph{state transition system (STS)} $\mathcal{A}=(S, f_0, f_1, T)$ consists of a set of states $S$, two transition functions $f_0,f_1:S\rightarrow S$, a starting state $s_0\in S$ and a set of accepting states $T\subseteq S$. The associated \emph{extended transition function} $\delta_{\mathcal{A}}:\{0,1\}^*\rightarrow S$ is defined as:
    \[\delta_{\mathcal{D}}(\epsilon):=s_0,\quad\delta_{\mathcal{D}}(l0):=f_0(\delta_{\mathcal{D}}(l)),\quad\delta_{\mathcal{D}}(l1):=f_1(\delta_{\mathcal{D}}(l)).\]
    When the associated STS is clear from context, we also omit the subscript when writing the extended transition function. An STS is \emph{finite} if the set of states $S$ is finite, and \emph{commutative} if $f_0(f_1(s)) = f_1(f_0(s))$ for all $s\in S$. We write \emph{CSTS} for commutative STS.
\end{definition}

\begin{lemma}
Let $\mathcal{A}=(S, f_0, f_1, T)$ be an STS. Then for any $l_1,l_1',l_2,l_2'\in \{0,1\}^*$, if $\delta(l_1)=\delta(l_1')$ and $\delta(l_2)=\delta(l_2')$, then $\delta(l_1l_2) = \delta(l_1'l_2')$.
\end{lemma}
\begin{proof}
\[\delta(l_1l_2) = \delta(l_1'l_2) = \delta(l_2l_1') = \delta(l_2'l_1') = \delta(l_1'l_2')\]
\end{proof}


\begin{definition}
Given a function $f:\{0,1\}^n\rightarrow\{0,1\}$, we say an STS \emph{describes $f$} if, for every $l\in \{0,1\}^n$, $\delta(l) \in T$ if and only if $f(l)=1$. Let $f,g:\{0,1\}^n\rightarrow\{0,1\}$ be two functions. We say that a CSTS  describes $f$ \emph{modulo literal-flipping} if it describes $g$ and there exists a literal-flipping function $\phi$ that sends $x_i$ to either $x_i$ or $\bar{x}_i$ such that $f=g \circ \phi$. We call the minimum number of states of a CSTS that describes~$f$ modulo literal-flipping the \emph{state size} of $f$. 
\end{definition}

\begin{figure}[H]
\centering
\begin{tikzpicture}[scale=1,main node/.style={circle,draw},end node/.style={red,thick,circle,draw}]
\node[main node](A) at (-1*0.6,0) {$s_0$};
\node[end node](B) at (1*0.6,0) {$s_1$};
\node(a) at (-1,1) {};
\draw[->,auto]
    (A) edge [bend left] node {1} (B)
    (B) edge [bend left]node {1} (A)
    (A) edge [out=135, in=225, loop] node {0} (A)
    (B) edge [out=45, in=-45, loop]node [left]{0} (B);
\draw[->] (a) edge node {} (A);
\end{tikzpicture}
\hspace{-9pt}
\begin{tikzpicture}[scale=1,main node/.style={circle,draw,font=\sffamily\bfseries},end node/.style={red,thick,circle,draw,font=\sffamily\bfseries}]
\node[main node](A) at (-1*0.6,0) {$s_0$};
\node[end node](B) at (1*0.6,0) {$s_1$};
\node(a) at (-1,1) {};
\draw[->,auto]
    (A) edge node {1} (B)
    (A) edge [out=135, in=225, loop] node {0} (A)
    (B) edge [out=45, in=-45, loop]node [left]{0,1} (B);
\draw[->] (a) edge node {} (A);
\end{tikzpicture}
\hspace{-8pt}
\begin{tikzpicture}[scale=1,main node/.style={circle,draw,font=\sffamily\bfseries},end node/.style={red,thick,circle,draw,font=\sffamily\bfseries}]
\node[main node](A) at (-1*0.6,0) {$s_0$};
\node[main node](B) at (1*0.6,0) {$s_1$};
\node[end node](C) at (3*0.6,0) {$s_2$};
\node(a) at (-1,1) {};
\draw[->,auto]
    (A) edge node {1} (B)
    (B) edge node {1} (C)
    (A) edge [out=135, in=225, loop] node {0} (A)
    (B) edge [loop] node {0} (B)
    (C) edge [loop] node {0,1} (C);
\draw[->] (a) edge node {} (A);
\end{tikzpicture}
\caption{CSTS for XOR, OR and ``$x_1 + \dots + x_n \geq 2$'' constraints.}
\label{STS-example}
\end{figure}
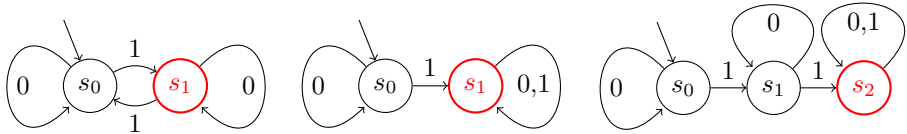

Every complete OBDD with $n$ variables can be reinterpreted as an STS $(\bigcup(S_i)_{1\leq i\leq n}, f_0, f_1)$, with the two sinks from the OBDD denoted as $\mathbf{0}_\mathcal{A}$ and $\mathbf{1}_{\mathcal{A}}$, or simply, $\mathbf{0}$ and $\mathbf{1}$, where $S_i$ is the set of decision nodes at the $i$th layer, $f_0(s),f_1(s)\in S_{i+1}$ for $s\in S_i, i<n$ and $f_0(s),f_1(s)\in \{\mathbf{0}, \mathbf{1}\}$ for $s\in S_n$. The concept of being commutative and the extended transition function $\delta$ transfer naturally.
\begin{lemma}\label{lemma:commutative}
Any symmetric function described by an OBDD of width $w$ can be described by an commutative OBDD of width at most $w$.
\end{lemma}

\begin{proof}
$\mathcal{D}=((S_i)_{1\leq i\leq n}, f_0, f_1, \mathbf{0}, \mathbf{1})$ be a commutative OBDD and let $r$ be its root. Suppose $n\leq 2$, then the property of being commutative follows vacuously. So it is safe to assume that $n\geq 3$. Define an equivalence relation $\sim_i$ over $S_i$ for each $1\leq i\leq n$. Define $\sim_1:=\{(r,r)\}$ and $\sim_2:=\{(f_0(r),f_0(r)),(f_1(r),f_1(r))\}$. Suppose $\sim_i$ and $\sim_{i+1}$ have already been defined. Define $\sim_{i+2}$ to be the transitive closure of $\{(f_0(x),f_0(y))\,|\,x\sim_{i+1}y\}\cup\{(f_1(x),f_1(y))\,|\,x\sim_{i+1}y\}\cup\{(f_0(f_1(x)),f_1(f_0(y)))\,|\,x\sim_iy\}$. Define $\widehat{\mathcal{D}}=((S_i/\sim_i)_{1\leq i\leq n}, \widehat{f_0}, \widehat{f_1},\mathbf{0}, \mathbf{1})$ by $\widehat{f_0}([x]):=[f_0(x)]$ and $\widehat{f_1}([x]):=[f_1(x)]$. It is easy to check that $\widehat{\mathcal{D}}$ indeed is a well-defined commutative OBDD that describes that same function as $\mathcal{D}$, and since $|S_i/\sim_i| \leq |S_i|$, the width of $\widehat{\mathcal{D}}$ is at most $w$.
\end{proof}

\begin{lemma}
\label{moduloPattern}
Let $f$ with $n$ inputs be a symmetric function described by a commutative OBDD of width $w$. Then there exist $0\leq a, b, m$ with $a+m+b=w$ such that $\delta(1^x0^y)=\delta(1^{x'}0^{y'})$ for any $0\leq x,x',y,y'$ such that $x+y=x'+y'=n$, $x,x'\geq a$, $y,y'\geq b$ and $x\equiv x'\text{ mod }m$.
\end{lemma}
\begin{proof}
Let $\mathcal{D}$ be a commutative OBBD of width $w$ that describes $f$. Then $n\geq w$. If $n=w$, then the claim is vacuously true with $m=0$ and any $0\leq a,b$, $a+b=w$. So we can assume $n>w$. By the pigeonhole principle, there exists $0\leq a,m,b$, $a+b+m=w$ such that $\delta(1^a0^m0^b)=\delta(1^a1^m0^b)$. Let $0\leq x,y$ such that $x+y=n$, $x\geq a + m$ and $y\geq b$. By commutativity we have 
$\delta(1^x0^y)=\delta(1^{a+(x-m-a)+m}0^{b+(y-b)})=\delta(1^a1^m0^b)\cdot\delta(1^{x-m-a}0^{y-b}) = \delta(1^a0^m0^b)\cdot\delta(1^{x-m-a}0^{y-b}) = \delta(1^{x-m}0^{y+m})$. Now the claim follows.
\end{proof}

\begin{lemma}
Let $f$ be a symmetric (respectively, literal-symmetric) function. If $f$ admits a complete OBDD representation of width $w$, then $f$ is described (described modulo literal-flipping) by a CSTS with at most $(w+1)^2/4$ states. 
\label{OBDDtoSTS}
\end{lemma}

\begin{proof}
Let $\mathcal{D}$ be a commutative OBDD and $0\leq a,m,b$ be such that $a+m+b=w$ and $\delta(1^x0^y)=\delta(1^{x'}0^{y'})$ for any $0\leq x,x',y,y'$ with $x+y=x'+y'=n$, $x,x'\geq a$, $y,y'\geq b$ and $x\equiv x'\text{ mod }m$. Define an STS $\mathcal{S}:=(S,f_0,f_1,s_0,T)$ where $S:=[0, b]\times[0, a+m-1]$, $f_0(i,j):=(i+1,j)$ if $i< b$, $f_0(b,j)=(b,j)$, $f_1(i,j):=(i,j+1)$ if $j< a+m-1$, $f_1(i,a+m-1):=(i,a)$ and $(i,j)\in T$ if and only if $\delta(1^i0^j)=\mathcal{1}$. It is easy to check that $\mathcal{S}$ and $\mathcal{D}$ indeed describes the same function. Note that since $a+m+b=w$, if follows that $(a+m)\cdot(b+1)\leq\frac{(w+1)^2}{4}$.

Let $f$ be a literal-symmetric function. Then there exists a literal-flipping $\phi$ such that $f=g\circ \phi$ for some symmetric function $g$. We have shown that there exists a CSTS with at most $\frac{(w+1)^2}{4}$ states that describes $g$. By definition, there exists a CSTS with at most $\frac{(w+1)^2}{4}$ states that describes $f$ modulo literal-flipping.
\end{proof}

\begin{figure*}
\centering
\[
\begin{matrix}
\text{functions with constant-} \\
\text{size one-sided CSTS } \\
\text{modulo literal-flipping}
\end{matrix}
\ \subset\ 
\begin{matrix}
\text{functions with} \\
\text{constant-size CSTS } \\
\text{modulo literal-flipping}
\end{matrix}
\ =\  
\begin{matrix}
\text{literal-symmetric} \\
\text{functions } O(1)\text{-slim}\\
\text{for OBDDs}
\end{matrix}
\ \subset\ 
\begin{matrix}
\text{functions } \\ 
O(1)\text{-slim}\\
\text{for OBDDs}
\end{matrix}
\begin{matrix}
\text{functions } \\ 
O(1)\text{-slim}\\
\text{for SDNNFs}
\end{matrix}
\]
\end{figure*}

\begin{table*}\def\arraystretch{1.5}

\begin{tabular}{ |c|c|c|c| } 
 \hline
 Name & Formal expression & $h$-slim for OBDDs & $h$-slim for SDNNFs\\
 \hline 
 \hline
 \tikzmark{C}Clauses & $x_1 \lor \dots \lor \bar x_n$ & $h = 2$ & $h = O(1)$\tikzmark{A} \\ 
 \hline
 XORs & $x_1 \oplus \dots \oplus x_n = 0$ or $1$ & $h = 2$ & $h = O(1)$ \\
 \hline 
 Sum modulo & $x_1 + \dots + x_n = c \mod k$  & $h = k$ & $h = O(k^2)$ \\ 
 \hline
 Cardinality & $x_1 + \dots + x_n \geq k$  & $h = O(\min(n-k,k))$ & $h = O(\min(n-k,k)^2)$ \tikzmark{B} \\ 
 \hline
 \makecell{Small scope \\ constraints} & $f(x_1,\dots,x_k)$ & $h = 2^k$ & $h = O(2^k)$ \\ 
 \hline
 \tikzmark{D}Threshold  & \makecell{$w_1x_1 + \dots + w_n x_n \geq \theta$ \\ with $k := |w_1| + \dots + |w_n|$} & $h = 2k+1$ & $h = O(k^2)$ \\ 
 \hline
 \makecell{Symmetric  \\ functions}  & \makecell{$f(x_1 + \dots + x_n)$ \\ for $f : \mathbb{N} \rightarrow \{0,1\}$} & $h = \poly(n)$ & $h = \poly(n)$ \\ 
 \hline
 \makecell{Literal-sym. \\ functions}  & \makecell{$f(\ell_1 + \dots + \ell_n)$ \\ for $f : \mathbb{N} \rightarrow \{0,1\}$ \\ and $\ell_i \in \{x_i, \bar x_i\}$} & $h = \poly(n)$ & $h = \poly(n)$ \\ 
 \hline
\end{tabular}
\caption{Families of constraints}\label{table:constraints}
\end{table*}

\section{FPT Compilation of Systems of $O(1)$-Slim Constraints}\label{section:fpt_compilation}

In this section, we explain the proof of Theorem~\ref{theorem:main_fpt_compile}. The proof is in three steps:
(1)~find a CNF encoding $H(X,Z)$ of $F(X)$ whose incidence treewidth is at most $O(w\cdot tw_i(F))$,
(2)~compile $H(X,Z)$ to d-SDNNF using
Theorem~\ref{theorem:cnf_fpt_compile}, and
(3)~existentially forget the auxiliary variables $Z$ from the resulting d-SDNNF circuit.
We do not have to take care of (2). For (3), \emph{forgetting} the
variables $Z$ from a d-SDNNF circuit $D(X,Z)$ means finding another
d-SDNNF circuit computing $\exists Z.D(X,Z)$. Forgetting many variables from a d-SDNNF circuit is generally
intractable~\cite{PipatsrisawatD08}. However the operation is tractable when the  $Z$-variables are completely defined in terms of the $X$-variables. We make sure to be in this situation by using only Tseitin CNF encodings to generate $H(X,Z)$.

\subsection{The Tseitin Encoding}\label{section:tseitin_encoding}

Let $D(X)$ be a Boolean circuit whose internal gates are binary $\lor$-gates and binary $\land$-gates. Let $gates^*(D)$ be the set of its internal gates. Every gate $g$ is associated with a variable $z_g$. Let $Z = \{ z_g \mid g \in \gates^*(D)\}$. The Tseitin encoding of~$D$ is the CNF
\[
H^\tseitin_D(X,Z) := \bigwedge\nolimits_{g \in \gates^*(D)} \tseitin(g)
\] where, for every gate $g$, if $g = g_1 \land g_2$ then
$
\tseitin(g) := (\bar z_g \lor z_{g_1}) \land (\bar z_g \lor z_{g_2}) \land (\bar z_{g_1} \lor \bar z_{g_2} \lor z_g)
$, if $g = g_1 \lor g_2$ then
$
\tseitin(g) := (\bar z_g \lor z_{g_1} \lor z_{g_2}) \land (\bar z_{g_1} \lor z_g) \land (\bar z_{g_2} \lor g)
$, and if $g$ is an input gate, so if $g = x,\bar x,0$ or $1$, then $z_g = g$. The formula $H^\tseitin_D$ is a CNF encoding of $D$ in the sense that 
\[
\exists Z. H^\tseitin_D(X,Z) \equiv D(X).
\]
\noindent We often drop the $\tseitin$ superscript in the rest of the paper. $\exists Z.H_D(X,Z)$ is the Boolean function over $X$ that maps to $1$ exactly the assignments to $X$ that can be extended over $Z$ to satisfying assignments of $H_D$. For $C$ a DNNF circuit and $Z$ a subset of $var(C)$, $\textsc{Exist}(C,Z)$ is the circuit $C$ where all literals in $C$ for $Z$ are replaced by $1$. Clearly, computing $\textsc{Exist}(C,Z)$ is a linear-time procedure.
 $\textsc{Exist}(C,Z)$ is in DNNF and is equivalent to $\exists Z.C$~\cite{DarwicheM02}. It is known that, when $C$ is a d-DNNF circuit representing a Tseitin encoding $H_D(X,Z)$, with $Z$ the Tseitin variables, then $\textsc{Exist}(C,Z)$ is not just in DNNF but in d-DNNF. This is not straightforward since the $\textsc{Exist}$ procedure generally breaks determinism for $\lor$-gates~\cite{KieselE23}. One can go one step further and prove that if $C$ is in SDNNF, then $\textsc{Exist}(C,Z)$ is also in SDNNF.

\begin{lemma}\label{lemma:forget_Tseitin_variables}
Let $f(X)$ be a Boolean function and $D_1(X_1),\dots,D_m(X_m)$ be Boolean circuits. For $D$ an d-SDNNF circuit computing $f(X) \land H^\tseitin_{D_1}(X_1,Z_1) \land \dots \land H^\tseitin_{D_m}(X_1,Z_m)$, where the sets $Z_1,\dots,Z_m$ are pairwise disjoint, the linear-time procedure $\textsc{Exist}(D, Z_1 \cup \dot \cup Z_m)$ returns a d-SDNNF circuit computing $f(X) \land D_1(X_1) \land \dots \land D_m(X_m)$.
\end{lemma} 
\begin{proof}
Let $Z = Z_1 \cup \dots \cup Z_m$. We have that $\textsc{Exist}(D,Z) \equiv \exists Z. D(X,Z)$ by definition of $\textsc{Exists}$, and $\exists Z.D(X,Z) \equiv f(X) \land \bigwedge_{i = 1}^m \exists Z_i. H^\tseitin_{D_i}(X_i,Z_i)$ because the sets $X$, $Z_1$, $\dots$, $Z_m$ are pairwise disjoint, and therefore $\textsc{Exist}(D,Z) \equiv f(X) \land \bigwedge_{i  = 1}^m D_i(X_i)$ by property of the Tseitin encoding.

Let $D' = \textsc{Exist}(D,Z)$. The proof that $D'$ is in d-DNNF can be found in~\cite{KieselE23}. So we just have to show that this d-DNNF circuit is structured by a vtree. Let $(\tau,\lambda)$ be a (vtree,mapping) pair that structures $D$. Then $D'$ is also structured by $(\tau,\lambda)$. Indeed, the definition of an SDNNF does not require that $\tau$ is exactly over the variables of the circuit it structures. $\tau$ can structure $D'$ as long as it is over a superset of $var(D')$. Recall that there are three conditions for $(\tau,\lambda)$ to structure $D'$. The two conditions on the $\land$-gates and $\lor$-gates are satisfied because the gates of $D'$ are the same as the gates of $D$. The condition that, for every $g \in gates(D')$, $var(D'_g) \subseteq var(\lambda(g))$ is also satisfied because $var(D'_g) \subseteq var(D_g)$ and $var(D_g) \subseteq var(\lambda(g))$. So $D'$ is structured by $(\tau,\lambda)$.
\end{proof}

\subsection{Proof of Theorem~\ref{theorem:main_fpt_compile}}\label{section:proof_fpt_compile}

The only real hurdle in the proof of Theorem~\ref{theorem:main_fpt_compile} is~(1), that is, encoding the system of constraints $F(X) = c_1(X_1) \land \dots \land c_m(X_m)$ into the CNF $H(X,Z)$ while controlling the incidence treewidth. For every $i \leq m$ we do the following. First, using a certain tree decomposition of $F$'s incidence graph, we construct a vtree $\tau_{c_i}$ over $X_i$. Second, we construct an SDNNF circuit $D_i(X_i)$ structured by $\tau_{c_i}$ that computes $c_i(X_i)$. Next, we encode $D_i$ into a CNF formula $H_i(X_i,Z_i)$ using the Tseitin encoding and a set $Z_i$ of fresh auxiliary variables. The CNF encoding of $F$ is then $H(X,Z) := H_1(X_1,Z_1) \land \dots \land H_m(X_m,Z_m)$. Given the vtrees $\tau_{c_i}$, it should be quite clear that $H(X,Z)$ can be constructed in polynomial time under Assumption~\ref{assumption}, and that the $Z$ variables can be forgotten using Lemma~\ref{lemma:forget_Tseitin_variables}. So we only need to justify that the $\tau_{c_i}$ can be found efficiently, and that the incidence treewidth of $H(X,Z)$ can be controlled.

\subsubsection{Constructing the Vtrees} To construct the vtree $\tau_c$ for $c \in F$, we need a tree decomposition (t.d.) of $G_F$ with specific properties. We do not give the details here but, basically, we use nice t.d.~in which a few bags are cloned. Roughly put, $\tau_c$ shows how the variables of $c$ appear relative to each other in the t.d. For instance, in the t.d.~shown Figure~\ref{figure:td_G_F} (not of the type used in the proof, but sufficient for the example) for the incidence graph of Figure~\ref{figure:system}, we have $x_2,x_3$ and $x_4,x_5$ are introduced in a different branches, and this yields the vtree $\tau_{c_2}$ shown Figure~\ref{figure:SDNNF}. $\tau_c$ is found using the following lemma.

\begin{lemma}\label{lemma:stepOne}
Every tree decomposition (t.d.)~of $G_F$ can be transformed in polynomial time into another t.d.~$(T,b)$ of $G_F$, that has the same width, and that can be used to find in linear time a vtree $\tau_c$ over $var(c)$ for every $c \in F$. Each vtree $\tau_c$ has a t.d.~$(T,b_c)$ of width $3$ such that, for all $t \in V(T)$,
\begin{itemize}
\item[$i.$] if $c \not\in b(t)$ then $b_c(t) = \emptyset$,
\item[$ii.$] if $c \in b(t)$ then $b_c(t) \cap b(t) \subseteq b(t) \cap var(c)$.
\end{itemize}
\end{lemma}
\begin{proof}
First, we turn the initial tree decomposition into a nice tree decomposition $(T',b')$, with same width, in polynomial time. We do not describe this step, it is well-known. 

In a tree, We say that we \emph{bypass} a node $t$ that has a single child $t'$ when we remove $t$ if it has no parent (so if $t$ is the root), or, if $t$ has a parent $t''$, when we remove $t$ and make $t''$ the parent of $t'$. We say that we \emph{insert a parent before} $t$ when, if $t$ is the root we create a new root whose unique child is $t$, and otherwise we disconnect $t$ from its parent $t'$ and create a new node $t''$ whose parent is $t'$ and whose unique child is $t$. We \emph{insert a clone parent} in a tree decomposition $(T,b)$ when we insert a parent $t''$ before a node $t$ and set $b(t'') = b(t)$. 

For every $x \in var(F)$, if the highest node $t$ of $T'$ such that $x \in b'(t)$ is a join node, then we insert a clone parent before $x$. Once this is done, for every $x$, let $t_x$ be the highest node of $T'$ such that $x \in b(t)$. The properties of nice tree decomposition guarantee that $t_x \neq t_y$ holds for every $y \neq x$ and the clone insertions guarantee that $t_x$ has zero or a single child.

Let $c \in F$ and let $S = var(c) \cap b'(t_c)$. For every $x \in S$, insert a clone parent before $t_c$. Call $t^x_c$ this new parent. For every $x \in var(c) \setminus S$, we have that $t_x$ is a descendant of $t_c$ and that $c \in b'(t_x)$. For every such $x$, $t^x_c$ refers to $t_x$. Repeat the construction for every $c \in F$. The resulting tree decomposition is $(T,b)$. It is clear that this is still a tree decomposition of $G_F$, that it has the same width as the decomposition we started with, and that it is constructed in polynomial time. Note that for every $x,y \in var(c)$, $t^x_c \neq t^y_c$ and that $t^x_c$ has zero or a single child. Now we explain how to construct $\tau_c$ for any $c \in F$.
\begin{itemize}
\item[•] Remove every node of $T$ that neither belongs to $\{t^x_c \mid x \in var(c)\}$, nor has a descendant in $\{t^x_c \mid x \in var(c)\}$.
\item[•] Remove every node of $T$ that neither belongs to $\{t^x_c \mid x \in var(c)\}$, nor has an ancestor in $\{t^x_c \mid x \in var(c)\}$.
\item[•] At that point, $T$ has been transformed into $T_1$ Repeatedly bypass all nodes not in $\{t^x_c \mid x \in var(c)\}$ that have a single child until reaching a fix point.
\item[•] At that point, $T$ has been transformed into $T_2$. For every $t^x_c$, if it is a leaf then label it by $x$; if it has one child $t'$ then give $t^x_c$ another child labeled by $x$. Since $t^x_c$ has zero or a single child in $T$, this is also the case in $T_ 2$.
\end{itemize}
The resulting tree is $\tau_c$. Let $V_c = \{t^x_c \mid x \in var(c)\}$. Note that $V_c \subseteq V(T_2) \subseteq V(T_1) \subseteq V(T)$. We will see the nodes of $\tau_c$ as distinct from $T$'s nodes though. Formaly, we use a function $\sigma$ defined as follows: if $t = t^x_c$ is a leaf of $T_2$, then $\sigma(t) = x$ is, and if $t$ is an internal node of $T_2$ then the corresponding node in $\tau_c$ is $\sigma(t)$ (note that in this case $\sigma(t) \neq x$ even when $t = t^x_c$).

The bag function $b_c$ is defined following the construction: every node $t$ removed during the fist two steps gets $b_c(t) = \emptyset$. These are all $t \in V(T) \setminus V(T_1)$. In particular every node $t$ such that $c \not\in b(t)$ is in $V(T) \setminus V(T_1)$ so $c \not\in b(t) \Rightarrow b_c(t) = \emptyset$ holds. By construction the root of $T_1$ and the leaves of $T_1$ belong to $V_c$. Every node $t$ in $V(T_1)$ is either in $V(T_2)$ of has a unique ancestor $a_t$ in $V(T_2)$ and a unique descendant $d_t$ in $V(T_2)$, in which case $a_t$ is the parent of $d_t$ in $T_2$. For $t \in V(T_1) \setminus V(T_2)$ we let $b_c(t) = \{\sigma(a_t),\sigma(d_t)\}$. For $t \in V(T_2) \cap V_c$ with parent $t'$ and such that $t = t^x_c$, we set $b_c(t) = \{\sigma(t),x,\sigma(t')\}$ (where $\sigma(t)$ may be $x$ if $t$ is a leaf of $T_2$). For $t \in V(T_2) \setminus V_c$ with parent $t'$, we set    $b_c(t) = \{\sigma(t),\sigma(t')\}$. 

One can verify that $(T,b_c)$ is a tree decomposition of $\tau_c$. Its width is at most $3$, and the bags $b_c(t)$ intersect $b(t)$ only on $var(c)$. So $b_c(t) \cap b(t) \subseteq b(t) \cap var(c)$ holds. 
\end{proof}

It may seem bizarre that Lemma~\ref{lemma:stepOne} looks at t.d.~for vtrees. The key point here is that the vtrees have width-$2$ t.d.~with \emph{the same underlying tree} $T$ as the t.d.~of $G_F$. For instance, a t.d.~$(T,b_{c_2})$ for $\tau_{c_2}$ from Figure~\ref{figure:SDNNF} using the same tree $T$ as the t.d.~from Figure~\ref{figure:td_G_F} could use $b_{c_2}(t_1) = \{\circled{1}{green}\}$, $b_{c_2}(t_2) = \{\circled{1}{green},\circled{3}{blue}\}$, $b_{c_2}(t_3) = \{\circled{3}{blue},x_4\}$, $b_{c_2}(t_4) = \{\circled{3}{blue},x_5\}$, $b_{c_2}(t_7) = \{\circled{1}{green},\circled{2}{red}\}$, $b_{c_2}(t_8) = \{\circled{2}{red},x_3\}$, $b_{c_2}(t_9) = \{\circled{2}{red},x_2\}$ and all the other bags empty.

\subsubsection{Controlling the Treewidth of the Encoding} Thanks to the properties stated in points $i.$ and $ii.$ of Lemma~\ref{lemma:stepOne}, we can \emph{merge} the t.d.~$(T,b_c)$ of $\tau_c$ with the t.d.~$(T,b)$ of $G_F$ with a simple bag-wise union to get a t.d.~of $G_F \cup \tau_c$ (the graph with vertex set $V(G_F) \cup V(\tau_c)$ and with edge set $E(G_F) \cup E(\tau_c)$).

\begin{lemma}\label{lemma:mergeTD}
Let $G$ and $G'$ be two graphs and let $V_s := V(G) \cap V(G')$.
Let $(T,b)$ and $(T',b')$ tree decompositions of $G$ and $G'$, respectively. If $T = T'$ and if for all $t \in V(T)$ we have $b(t) \cap V_s \subseteq b'(t) \cap V_s$ then $(T, b \cup b')$ is a tree decomposition of $G \cup G'$.
\end{lemma} 
\begin{proof}
We have that $\bigcup_{t  \in V(T)} (b\cup b')(t) = \bigcup_{t \in V(T)} b(t) \cup \bigcup_{t \in V(T)} b'(t) = V(G) \cup V(G') = V(G \cup G')$. Every edge of $G \cup G'$ is an edge of $G$ or an edge $G'$, so it is contained in $b(t)$ or $b'(t)$ for some $t \in V(T)$, a fortiori it is contained in $(b \cup b')(t)$. Now, for the connectivity condition, if $v \in V(G) \setminus V(G')$ then we have that $v \in (b \cup b')(t)$ if and only if $v \in b(t)$, so $T[ t \mid v \in (b \cup b')(t)] = T[ t \mid v \in b(t)]$ is a connected tree. The argument is similar for $v \in V(G') \setminus V(G)$. Finally, if $v \in V_s$ then there is a node $t$ such that $v \in b(t)$, but then $v \in b'(t)$ by assumption. It follows that the two trees $T[t \mid v \in b(t)]$ and $T[t \mid v \in b'(t)]$ share a vertex and therefore $T[t \mid v \in (b \cup b')(t)]$ is connected.
\end{proof}

If we merge $(T,b)$ with all the t.d.~$(T,b_{c_1}),\dots,(T,b_{c_m})$
obtained from Lemma~\ref{lemma:stepOne} then we obtain a
t.d.~$(T,b^*)$ of $G_F \cup \tau_{c_1} \cup \dots \cup \tau_{c_m}$ of
width at most $3\cdot \textit{width}(T,b)$ because, by property $i.$, new elements are added to a bag $b(t)$ only from the vtrees $\tau_c$ such that $c \in b(t)$, and each vtree contributes at most $3$ elements.

Now suppose we have a complete DNNF circuit $D_i$ computing $c_i$ and respecting the vtree $\tau_{c_i}$. Every gate of $D_i$ and, by extension, every clause and every auxiliary variable of its Tseitin encoding, is mapped to a node of $\tau_{c_i}$. We claim that, if for every bag $b^*(t)$ and every node $s$ of $\tau_{c_i}$ in this bag, one replaces $s$ by the clauses $\clauses(\tseitin(g))$ and the auxiliary variables used in these clauses for all gates $g$ of $D_i$ mapped to $s$, then the result is a t.d.~$(T,b^+)$ of $G_F \cup G_{H_1} \cup \dots \cup G_{H_m} = G_F \cup G_H$. One then just have to remove the initial constraints of $F$ from every bag to obtain a t.d.~of $G_H$. Going back to our example, if $D_2$ is the SDNNF circuit shown in Figure~\ref{figure:SDNNF}, one would replace all occurrences of the vtree node \circled{2}{red} in $b^*(t)$ by the clauses and variables used in the Tseitin encoding of the nodes squared in red in Figure~\ref{figure:SDNNF}: each such gate $g$ with input $g_1$ and $g_2$ will contribute the clauses of $\tseitin(g)$ plus the variables $z_g$, $z_{g_1}$ and $z_{g_2}$ when they are not input variables of $D_2$.

Since every $c_i$ is $w$-slim for complete SDNNFs, at most $w$ gates of $D_i$ is mapped to $s$, and the Tseitin encoding generates a constant number of clauses and auxiliary variables per gate of $D_i$. So we will have that $|b^+(t)| \in O(w|b^*(t)|) \in O(w|b(t)|)$. So, if we started from a t.d.~of $G_F$ of width $O(tw_i(F))$ (computable in time FPT parameterized by $\tw_i(F)$~\cite{Bodlaender96}), then we end up with an CNF encoding $H$ whose incidence treewidth is $O(w \cdot \tw_i(F))$. 

In the detailed proof, available in the long version of the paper, we directly merge the tree decompositions of the CNF encodings $H_i$, whose incidence treewidth we control using the following lemma for $D = D_i$ and $\tau = \tau_{c_i}$. This is equivalent to what we have just explained.

\begin{lemma}\label{lemma:stepThree}
Let $D(Y)$ be a complete SDNNF circuit respecting $(\tau,\lambda)$. Let $\calT = (T,b)$ be a tree decomposition of $\tau$ and let $\phi$ be defined as $\phi(x) = \{x\}$ for every variable $x$ and as $\phi(s) = \bigcup_{g:\lambda(g) = s} var(\tseitin(g)) \cup \clauses(\tseitin(g))$ for every internal node $s \in \tau$. Let $b'(t) = \bigcup_{s \in b(t)}\phi(s)$. Then $(T,b')$ is a tree decomposition of $H^\tseitin_D$'s incidence graph and has width $O(\mathit{width}(D)\mathit{width}(\calT))$.
\end{lemma}
\begin{proof}
We verify that the three requirements for $(T,b')$ to be a tree decomposition of $G_{H_D}$ are met.
\begin{itemize}
\item[•]First, we show that every clause and variable of $H_D$ appears in some bag $b'(t)$. Since $\bigcup_{t \in V(t)} b(t) = V(\tau)$ we have  
$$
\bigcup_{t \in V(t)} b'(t) = \bigcup_{t \in V(t)}\bigcup_{s \in b(t)}\phi(s) = \bigcup_{s \in V(\tau)} \phi(s) = \bigcup_{s \in V(\tau)} \bigcup_{g:\lambda(g) = s} var(\tseitin(g)) \cup \clauses(\tseitin(g))
$$
Since every gate of $H_D$ is mapped to some node of $\tau$ through $\lambda$. It follows that 
$$
\bigcup_{t \in V(t)} b'(t) = \bigcup_{g \text{ a gate of } D} var(\tseitin(g)) \cup \clauses(\tseitin(g)) = var(H_D) \cup \clauses(H_D)
$$
\item[•] Second, we look at the edges of $G_{H_D}$. Every edge in $G_{H_D}$ connects a variable from $var(\tseitin(g))$ to a clause from $\clauses(\tseitin(g))$. For every gate $g$, every $t \in V(T)$, every $c \in \clauses(\tseitin(g))$ and every $z \in var(\tseitin(g))$, we have that $c \in b'(t) \Rightarrow \clauses(\tseitin(g)) \cup var(\tseitin(g)) \subseteq b'(t) \Rightarrow z \in b'(t)$. So, for every edge of $G_{H_D}$, there is a bag that contains its two endpoints.
\item[•] Third, we look at the connectivity requirement. For the vertices that correspond to clauses this is easy: all clauses in $\clauses(\tseitin(g))$ appear precisely in the bags $b'(t)$ such that $\lambda(g) \in b(t)$, so we just have to use that $T[t \mid \lambda(g) \in b(t)]$ is a tree and we are done. Now let $z$ be a variable of $H_D$ such that $s = \lambda(z)$ is not the root of $\tau$, and let $s'$ be the parent of $s$ in $\tau$. By completeness of $D$, $z$ and $\neg z$ can only be inputs to gates $g$ such that $\lambda(g) = s'$. So, if $z \in var(\tseitin(g))$ then either $\lambda(g) = s$ or $\lambda(g) = s'$. It follows that $T[t \mid z \in b'(t)]$ is the union of $T[t \mid s \in b(t)]$ and $T[t \mid s' \in b(t)]$. Since $s$ and $s'$ are connected in $\tau$ there must be a node $t \in V(T)$ such that both $s \in b(t)$ and $s' \in b(t)$ hold, and therefore $T[t \mid z \in b'(t)]$ is connected. It remains the case where $z$ is a variable of $H_D$ such that $\lambda(z) = s$ is the root of $\tau$. Such variables appear in $var(\tseitin(g))$ if and only if $\lambda(g) = s$, so $T[t\mid z \in b'(t)]$ is exactly the tree $T[ t \mid s \in b(t)]$. 
\end{itemize}
So we have shown that $(T,b')$ is indeed a tree decomposition of $H_D$'s incidence graph. The bound on its width is clear:
\[
\begin{aligned}
\max_{t} |b'(t)| &= \max_t \Big\vert \bigcup_{s \in b(t)}\phi(s) \Big\vert \leq \max_t |b(t)|\cdot\max_{s \in b(t)} |\phi(s)| \leq \max_t |b(t)| \cdot \max_{s \in V(\tau)} |\phi(s)| \\ &\leq \mathit{width}(\calT)\cdot O(\mathit{width}(D))
\end{aligned}
\]
\end{proof}

\begin{figure*}
\centering
\begin{subfigure}[t]{0.5\textwidth}
\centering
\begin{tikzpicture}[xscale=0.72]
\node[inner sep=2pt] (a) at (1,0) {$c_1$};
\node[inner sep=2pt] (b) at (2.5,0) {$c_2$};
\node[inner sep=2pt] (c) at (4,0) {$c_3$};

\node[inner sep=2pt] (x1) at (0,1) {$x_1$};
\node[inner sep=2pt] (x2) at (1,1) {$x_2$};
\node[inner sep=2pt] (x3) at (2,1) {$x_3$};
\node[inner sep=2pt] (x4) at (3,1) {$x_4$};
\node[inner sep=2pt] (x5) at (4,1) {$x_5$};
\node[inner sep=2pt] (x6) at (5,1) {$x_6$};
\node[inner sep=2pt] (x7) at (6,1) {$x_7$};

\draw (a) -- (x1);
\draw (a) -- (x2);
\draw (a) -- (x3);

\draw (b) -- (x2);
\draw (b) -- (x4);
\draw (b) -- (x5);
\draw (b) -- (x3);

\draw (c) -- (x5);
\draw (c) -- (x4);
\draw (c) -- (x6);
\draw (c) -- (x7);

\node[font=\small] (text) at (3,2.25) {$
F = 
\begin{cases} 
c_1 : x_1 \lor x_2 \lor x_3
\\
c_2 : x_2 \oplus x_3 \oplus x_4 \oplus x_5 = 1
\\
c_3 : x_4 + x_5 + x_6 + x_7 \geq 2
\end{cases}
$};
\end{tikzpicture}
\caption{The incidence graph of $F$}\label{figure:system}
\end{subfigure}\begin{subfigure}[t]{0.5\textwidth}
\centering
\begin{tikzpicture}[yscale=0.75, xscale=0.5]
\def\s{1}
\node[thick,draw=green,inner sep=2*\s] (o1) at (4.5,4) {$\lor$};

\node[thick,draw=green,inner sep=2*\s]  (a1) at (2.5,3) {$\land$};
\node[thick,draw=green,inner sep=2*\s]  (a2) at (6.5,3) {$\land$};

\node[thick,draw=red,inner sep=2*\s] (o2) at (1.5,2) {$\lor$};
\node[thick,draw=red,inner sep=2*\s] (o3) at (3.5,2) {$\lor$};
\node[thick,draw=blue,inner sep=2*\s] (o4) at (5.5,2) {$\lor$};
\node[thick,draw=blue,inner sep=2*\s] (o5) at (7.5,2) {$\lor$};

\node[thick,draw=red,inner sep=2*\s] (a3) at (1,1) {$\land$};
\node[thick,draw=red,inner sep=2*\s] (a4) at (2,1) {$\land$};
\node[thick,draw=red,inner sep=2*\s] (a5) at (3,1) {$\land$};
\node[thick,draw=red,inner sep=2*\s] (a6) at (4,1) {$\land$};
\node[thick,draw=blue,inner sep=2*\s] (a7) at (5,1) {$\land$};
\node[thick,draw=blue,inner sep=2*\s] (a8) at (6,1) {$\land$};
\node[thick,draw=blue,inner sep=2*\s] (a9) at (7,1) {$\land$};
\node[thick,draw=blue,inner sep=2*\s] (a10) at (8,1) {$\land$};

\node[inner sep=\s] (x3) at (1,0) {$\,\,x_2$};
\node[inner sep=\s] (nx3) at (2,0) {$\,\,\bar x_2$};
\node[inner sep=\s] (x4) at (3,0) {$\,\,x_3$};
\node[inner sep=\s] (nx4) at (4,0) {$\,\,\bar x_3$};
\node[inner sep=\s] (x5) at (5,0) {$\,\,x_4$};
\node[inner sep=\s] (nx5) at (6,0) {$\,\,\bar x_4$};
\node[inner sep=\s] (x6) at (7,0) {$\,\,x_5$};
\node[inner sep=\s] (nx6) at (8,0) {$\,\,\bar x_5$};

\draw (a1) -- (o1) -- (a2);
\draw (a3) -- (o2) -- (a4);
\draw (a5) -- (o3) -- (a6);
\draw (a7) -- (o4) -- (a8);
\draw (a9) -- (o5) -- (a10);

\draw (o2) -- (a1) -- (o4);
\draw (o3) -- (a2) -- (o5);

\draw (x3.north) -- (a3); \draw (a3.south east) -- (x4.north);
\draw (nx3.north) -- (a4); \draw (a4.south east) -- (nx4.north);
\draw (x3.north) -- (a5.south west); \draw (a5) -- (nx4.north);
\draw (nx3.north) -- (a6.south west); \draw (a6.south) -- (x4.north);
\draw (x5.north) -- (a7); \draw (a7.south east) -- (x6.north);
\draw (nx5.north) -- (a8); \draw (a8.south east) -- (nx6.north);
\draw (x5.north) -- (a9.south west); \draw (a9) -- (nx6.north);
\draw (nx5.north) -- (a10.south west); \draw (a10.south) -- (x6.north);

\def\ox{10};
\def\oy{4};

\node[circle,draw=green,thick,inner sep=1pt,font=\small] (1) at (\ox,\oy) {1};
\node[circle,draw=red,thick,inner sep=1pt,font=\small] (2) at (\ox-1,\oy-0.75) {2};
\node[circle,draw=blue,thick,inner sep=1pt,font=\small] (3) at (\ox+1,\oy-0.75) {3};
\node[inner sep=1pt] (x2) at (\ox-1.5,\oy-1.5) {$x_2$};
\node[inner sep=1pt] (x3) at (\ox-0.5,\oy-1.5) {$x_3$};
\node[inner sep=1pt] (x4) at (\ox+0.5,\oy-1.5) {$x_4$};
\node[inner sep=1pt] (x5) at (\ox+1.5,\oy-1.5) {$x_5$};
\draw (2) -- (1) -- (3);
\draw (x2) -- (2) -- (x3);
\draw (x4) -- (3) -- (x5);
\node[font=\small] (label) at (\ox,\oy-2) {vtree $\tau_{c_2}$};
\end{tikzpicture}
\caption{An SDNNF circuit for $c_2$}\label{figure:SDNNF}
\end{subfigure}
\begin{subfigure}[t]{\textwidth}
\centering
\begin{tikzpicture}[xscale=0.9]

\node[label={below:\textit{root}},draw,rounded corners=0.1cm,inner sep=2.5]  (3) at (1.5,0) {$c_2\begin{matrix}\strut \\ \strut \end{matrix}$};
\node[draw,rounded corners=0.1cm,inner sep=2.5] (4) at (2.3,0) {$\begin{matrix}c_2 \\ c_3\end{matrix}$};
\draw (3) -- (4);

\node[draw,rounded corners=0.1cm,inner sep=2.5] (0) at (3.2,-0.75) {$\begin{matrix}c_2 \\ c_3\end{matrix}$};
\node[draw,rounded corners=0.1cm,inner sep=2.5] (5) at (4.25,-0.75) {$\begin{matrix}c_2\,x_3 \\ c_3 \end{matrix}$};
\node[draw,rounded corners=0.1cm,inner sep=2.5] (6) at (5.5,-0.75) {$\begin{matrix}c_2\,x_3 \\ x_2 \end{matrix}$};
\node[draw,rounded corners=0.1cm,inner sep=2.5] (7) at (6.75,-0.75) {$\begin{matrix}c_1\,x_3 \\ x_2  \end{matrix}$};
\node[draw,rounded corners=0.1cm,inner sep=2.5] (8) at (7.8,-0.75) {$\begin{matrix}c_1 \\ x_1 \end{matrix}$};
\draw[draw,rounded corners=0.1cm,inner sep=2.5] (4) -- (0) -- (5) -- (6) -- (7) -- (8);

\node[fill=white,text=gray,circle,inner sep=0,font=\small] (b) at (1.765,0.5) {$t_0$};
\node[fill=white,text=gray,circle,inner sep=0,font=\small] (b) at (2.57,0.5) {$t_1$};

\node[fill=white,text=gray,circle,inner sep=0,font=\small] (b) at (3.5,-.25) {$t_7$};
\node[fill=white,text=gray,circle,inner sep=0,font=\small] (b) at (4.75,-.25) {$t_8$};
\node[fill=white,text=gray,circle,inner sep=0,font=\small] (b) at (6,-.25) {$t_9$};
\node[fill=white,text=gray,circle,inner sep=0,font=\small,text width=2.5mm] (b) at (7.2,-.25) {$t_{10}$};
\node[fill=white,text=gray,circle,inner sep=0,font=\small] (b) at (8.1,-.25) {$t_{11}$};

\node[draw,rounded corners=0.1cm,inner sep=2.5] (0) at (3.2,+0.75)  {$\begin{matrix}c_2 \\ c_3\end{matrix}$};
\node[draw,rounded corners=0.1cm,inner sep=2.5] (5) at (4.35,+0.75)  {$\begin{matrix}c_2\,x_4 \\ c_3 \end{matrix}$};
\node[draw,rounded corners=0.1cm,inner sep=2.5] (6) at (5.7,+0.75)  {$\begin{matrix}c_2\,x_5 \\ c_3 \end{matrix}$};
\node[draw,rounded corners=0.1cm,inner sep=2.5] (7) at (6.825,+0.75) {$\begin{matrix}c_3 \\ x_6 \end{matrix}$};
\node[draw,rounded corners=0.1cm,inner sep=2.5] (8) at (7.7,+0.75) {$\begin{matrix}c_3 \\ x_7 \end{matrix}$};
\draw (4) -- (0) -- (5) -- (6) -- (7) -- (8); 

\node[fill=white,text=gray,circle,inner sep=0,font=\small] (b) at (3.5,1.25) {$t_2$};
\node[fill=white,text=gray,circle,inner sep=0,font=\small] (b) at (4.85,1.25) {$t_3$};
\node[fill=white,text=gray,circle,inner sep=0,font=\small] (b) at (6.18,1.25) {$t_4$};
\node[fill=white,text=gray,circle,inner sep=0,font=\small] (b) at (7.12,1.25) {$t_5$};
\node[fill=white,text=gray,circle,inner sep=0,font=\small] (b) at (8,1.25) {$t_6$};
\end{tikzpicture}
\caption{A tree decomposition of $G_F$}~\label{figure:td_G_F}
\end{subfigure}
\caption{}
\end{figure*}

\section{Lower Bounds for Systems of Slim Constraints}\label{section:lower_bounds}

Theorem~\ref{theorem:main_fpt_compile} requires a compilation time with a $2^{O(w\cdot k)}$ component, with $k = \tw_i(F)$ and $w$ an upper bound on the width of complete SDNNF circuits representing the constraints. Can we get rid of $w$ in the exponent? The
answer is negative due to
Lemma~\ref{lemma:slimSDNNF_do_not_give_slimOBDD}. One can use the hard
functions from this lemma as constraints and show that
$w$ cannot be dropped from the exponent in the compilation time simply by considering systems made of a single constraint (so of incidence treewidth $1$). The hard functions of
Lemma~\ref{lemma:slimSDNNF_do_not_give_slimOBDD} are specific monotone
DNF formulas and the proof uses d-SDNNF lower bounds shown by~\cite{AmarilliCMS20}. But DNF formulas can be converted in linear time in SDNNF circuits, so the hard functions for  Lemma~\ref{lemma:slimSDNNF_do_not_give_slimOBDD} admit SDNNF circuits of width polynomial in the number $n$ of variables, whereas they only have OBDD representations of width exponential in $n$. Thus it could be possible that the best running time of an FPT compilation to d-SDNNF is of the form $f(w+k)poly(n+w)$ when $w$ is
the SDNNF-width, but also of the form $f(k)poly(n+w)$ when $w$ is the
OBDD-width. We can prove that, even for classes of functions that
are $w$-slim for OBDDs, we cannot remove $w$ from $f$'s
argument.

\begin{theorem}\label{theorem:lowerBound}
For every $k$, there exist systems of constraints over $n$ variables and of incidence treewidth $O(k)$ whose d-SDNNF representations all have size $(n/k)^{\Omega(k)}$, whereas these systems only comprise constraints that are $O(nk)$-slim for OBDDs. 
\end{theorem}
The
hard systems are CSP (constraints satisfaction problems) encodings of
W[1]-hard problems. It seems unlikely that there exist FPT-reductions parameterized by incidence treewidth of W[1]-hard problems to CSP where all constraints are
$O(1)$-slim for complete OBDDs. Indeed, using Theorem~\ref{theorem:main_fpt_compile} plus the fact that deciding the satisfiability of a d-SDNNF circuit is straightforward, the existence of such reductions would imply FPT = W[1]. So, to find hard systems of constraints for the theorem, we started from a W[1]-hard problem with parameter $k$ and reduced it to a CSP of incidence treewidth $O(k)$ and whose constraints are all $O(nk)$-slim for OBDDs (and not
$O(1)$-slim). We have used the problautem $k$-\textsc{Clique}. The reduction is inspired by that used by \cite[Theorem 6]{SamerS10}for $k$-\textsc{IndependentSet}.

\subsection{The Hard Systems of Constraints}

The system of constraints is an encoding of a W[1]-hard problem: $k$-\textsc{Clique}.\medskip

\noindent\fbox{\parbox{0.98\linewidth}{$k$-\textsc{Clique}\\
\strut$\quad$ Input: a graph $G$ \\
\strut$\quad$ Parameter: an integer $k$ \\
\strut$\quad$ Output: \textit{yes} if $G$ has a clique of size $k$, \textit{no} otherwise
}}\\

The hard systems of constraints take two parameters: a graph $G$ and an integer $k$. For every $i \in \{1,\dots,k\}$ and every $\{u,v\} \in E(G)$, we introduce two Boolean variables $x_i[uv]$ and $x_i[vu]$. For conciseness, we write $[k] = \{1,\dots,k\}$. For $u \in V(G)$, we denote by $N(u)$ its set of neighbors in $G$. An \emph{ordered $k$-clique} of $G$ is a sequence $k$-clique of $G$ whose vertices are numbered from $1$ to $k$. Formally, it is represented as a sequence of distinct vertices $S = (u_1,\dots,u_k)$. We write $\elements(S) = \{u_1,\dots,u_k\}$. For $1 \leq \ell \leq h \leq k$, we denote by $S[\ell,h]$ the subsequence $(u_\ell,\dots,u_h)$. We will also denote by $[\ell,h]$ the set $\{\ell,\ell+1,\dots,h\}$.

Utimately, the satisfying assignments of the systems of constraint $F_{G,k}$ will be in bijection with the ordered $k$-cliques of $G$. The satisfying assignment $\alpha$ for a given ordered $k$-clique will be such that $\alpha(x_i[uv]) = 1$ if and only if $u$ is the $i$th vertex of a $k$-clique of $G$ and $v \in N(u)$ is also in the clique. We define two types of constraints.\\

The constraint $\chi_i$, $1 \leq i \leq k$, states that there is a vertex $u \in V(G)$ that is the $i$th vertex of the ordered $k$-clique.
\[
\chi_i = \bigvee_{u \in V(G)} \Big(\big(\sum_{v \in N(u)} x_i[uv] = k-1\big) 
 \land \big(\sum_{vw \in E(G)} x_i[vw]+x_i[wv] = k-1 \big)\Big)
\]
The constraint $\chi_{uv}$, $\{u,v\} \in E(G)$, states that  either $uv$ is contained in the ordered $k$-clique, in which case either $x_i[uv]$ or $x_i[vu]$ is set to $1$ for a unique $i$, or $uv$ is not contained in the clique and both $x_i[uv]$ and $x_i[vu]$ are set to $0$.
\[
\chi_{uv} = \Big(\big(\sum_{i \in [k]}x_i[uv] = 1\big)\land  \big(\sum_{i \in [k]} x_i[vu] = 1\big)\Big) \lor \big(\bigwedge_{i \in [k]} (x_i[uv] = x_i[vu] = 0)\big)
\]

Let 
\[
F_{G,k} = \bigwedge_{i \in [k]} \chi_i \land \bigwedge_{\{u,v\} \in E(G)} \chi_{uv}
\]

\begin{lemma}\label{lemma:models_are_ordered_k_cliques}
Let $S$ be an ordered $k$-clique and $\alpha_S$ be the assignment defined by $\alpha(x_i[uv]) = 1$ if and only if $u$ is the $i$th vertex of $S$ and $v \in N(u) \cap S$. An assignment to $\{x_i[uv] \mid i \in [k], \{u,v\} \in E(G)\}$ satisfies $F_{G,k}$ if and only if it is $\alpha_S$ for some ordered $k$-clique of $G$.
\end{lemma}
\begin{proof}
It is not too hard to verify that $\alpha_S$ indeed satisfies all functions $\chi_i$ and $\chi_{uv}$ so we focus on proving that the models of $F_{G,k}$ are all of the form $\alpha_S$ for some ordered clique $S$.

Let $X := \var(F_{G,k})$ and let $\alpha : X \rightarrow \{0,1\}$ be a model of $F_{G,k}$. Since the disjuncts of $\chi_i$ are pairwise inconsistent, we have that for every $i \in [k]$, there is a unique $u_i \in V(G)$ such that $\alpha(x_i[u_iv]) = 1$ for some $v$. But for every $v \in N(u)$ such that $\alpha(x_i[u_iv]) = 1$, $\alpha$ can only satisfy $\chi_{u_iv}$ if $\sum_{i \in [k]}\alpha(x_i[u_iv]) = 1$ and $\sum_{i \in [k]}\alpha(x_i[vu_i]) = 1$. So if there is a $j \neq i$ such that $\alpha(x_j[vu_i]) = 1$, so if $v = u_j$. Since $\alpha$ satisfies $\chi_i$, there must $k-1$ distinct such $v$ in $N(u_i)$, so the $u_j$ for $j \neq i$ are all neighbors of $u_i$ and are pairwise distinct. Applying this reasoning for every $i \in [k]$, we derive that the set $\{u_1,\dots,u_k\}$ is a $k$-clique of $G$, so $\alpha = \alpha_{(u_1,\dots,u_k)}$.
\end{proof}

We prove that the $\chi_i$ and $\chi_{uv}$ are $O(nk)$-slim for complete OBDD and that the incidence treewidth of $F_{G,k}$ is at most $k$.

\begin{lemma}\label{lemma:chi_uv_are_slim_for_complete_OBDD}
$\chi_{uv}$ can be represented by a complete OBDD of width at most $8$ for every variable ordering.
\end{lemma}
\begin{proof}
Two OBDD $B_1$, $B_2$ respecting the same variable ordering $\pi$ can be used to create another OBDD respecting $\pi$ and representing $B_1 \land B_2$ (resp. $B_1 \lor B_2$) and whose width is at most that of $B_1$ times that of $B_2$~\cite[Chapter 3]{Wegener00}. Using this, since the three constraints $\sum_{i \in [k]}x_i[uv] = 1$, $\sum_{i \in [k]}x_i[vu] = 1$ and $\sum_{i \in [k]}x_i[uv]+x_i[vu] = 0$ admit representations as complete OBDD of width $2$, for every variable ordering, it follows that $\chi_{uv}$ admit representations as complete OBDD of width at most $8$, for every variable ordering.
\end{proof}

\begin{lemma}\label{lemma:chi_i_are_slim_for_complete_OBDD}
$\chi_{i}$ can be represented by a complete OBDD of width at most $nk+2$.
\end{lemma}
\begin{proof}
Let $\pi$ be an ordering of $X := \var(F_{G,k})$. Consider the $j$th level of a complete $\pi$-OBDD computing $\chi_i$ (level $1$ is the level containing the source node). Let $X_{j-1}$ be the $j-1$ first variables in $\pi$ and let $X_{u} := \{x_i[uv] \mid v \in N(u)\}$. The OBDD only requires $nk+2$ nodes at level $j$: one node corresponding to the assignments that satisfy $\sum_{x \in X_{j-1}} x = 0$, another node corresponding to the assignments that satisfy $x_i[uv] = x_i[w\omega] = 1$ for some $u \neq w$ (these assignments already falsify $\chi_i$) and the remaining $nk$ nodes correspond to the situations where ``only variables $x_i[u\ast]$ have been assigned to $1$ for some $u \in V(G)$, and $(\sum_{x \in X_{j-1} \cap X_u}x)$ equals $1$, $\dots$, $k-1$, or is $> k-1$''. So the width of the OBDD is at most $nk+2$.
\end{proof}

\begin{lemma}\label{lemma:tw_at_most_k}
The treewidth of the incidence graph of $F_{G,k}$ is at most $k$.
\end{lemma}
\begin{proof}
We show that the \emph{dual} treewidth of $F_{G,k}$ is at most $k$ and then use that the incidence treewidth of a constraint system is at most its dual treewidth plus $1$~\cite{SamerS10}. The dual graph of $F_{G,k}$ is the graph whose vertices are its constraints and where two constraints are linked by an edge if and only if they share a variable. The variable sets of the $\chi_i$s are pairwise disjoint and the variable sets of the $\chi_{uv}$ are pairwise disjoint. But the scope of every $\chi_i$ intersect the scope of every $\chi_{uv}$, so the dual graph of $F_{G,k}$ is the complete bipartite graphe for the partition $(\{\chi_i \mid i \in [k]\}, \{\chi_{uv} \mid uv \in E(G) \})$. It is easy to see that a possible tree  decomposition of is a path comprising $|E(G)|$ nodes, one per edge $uv$ of $G$, with the bags $\{\chi_{uv}, \chi_1,\dots,\chi_k\}$. This tree decomposition has width $k+1$, so the dual treewidth of $F_{G,k}$ is at most $k$.
\end{proof}

\subsection{The DNNF lower bound}

To prove Theorem~\ref{theorem:lowerBound}, we show that, when $G$ is the complete graph on $n$ vertices, $F_{G,k}$ can only be computed by DNNF of size $(n/k)^{\Omega(k)}$. The whole subsection is dedicated to prove this. We use a lower bound technique based on the adversarial game construction of a rectangle covering of the $F_{G,k}$. There are a few notions to unpack here, starting with rectangles.

\begin{definition}[Rectangle]
A rectangle over a set $X$ of Boolean variables with respect to a bipartition $(X_1,X_2)$ of $X$ is a Boolean function $r(X) = \rho_1(X_1) \land \rho_2(X_2)$ where $\rho_1$ and $\rho_2$ are Boolean functions.
\end{definition}

\begin{definition}[Rectangle Cover]
A \emph{rectangle cover} of a Boolean function $f(X)$ is a collection $R = \{r_1(X),\dots,r_s(X)\}$ of rectangles, possibly w.r.t. different partitions, such that $f \equiv \bigvee_{i = 1}^s r_i(X)$.
\end{definition}

Rectangle covers exist for every Boolean function $f$. We are interested in the covers of $f$ obtained by an adversarial two-players game. In this game, Charlotte wants to minimize the size of the cover (i.e., the number of rectangles) while Adam tries to maximize it. A single round of the game unfold as follows:
\begin{itemize}
\item[1.] Charlotte chooses an assignment $\alpha \in f^{-1}(1)$ not already covered by $R$ and a vtree $\tau$ over $X$. 
\item[2.] Adam chooses a cut of $\tau$ that induces a partition $(X_1,X_2)$ of $X$.
\item[3.] Charlotte constructs a rectangle $r(X)$ with respect to $(X_1,X_2)$ such that $\alpha \in r^{-1}(1) \subseteq f^{-1}(1)$, and adds it to $R$.
\end{itemize}
Each round adds a rectangle to the cover. The game ends when all models of $f$ are covered. $R_a(f)$ is the size of the smallest cover obtained playing this game.

\begin{theorem}[\cite{deColnetM23}]\label{theorem:adversarialLowerBound}
Let $f$ be a Boolean function, the size of every complete DNNF circuit computing $f$ is at least $2^{R_a(f)}$.
\end{theorem}

Given a vtree $\tau$ over $X := var(F_{G,k})$ let $X_t$ be the set of variables on leaves below $t$ and $\overline{X}_t := X \setminus X_t$. A partition $(X_t,\overline{X}_t)$ (or $(\overline{X}_t,X_t)$) is said to be \emph{induced by $\tau$}. 

\begin{lemma}\label{lemma:goodCut}
Let $\tau$ be a vtree over $X$ and let $Y \subseteq X$ with $|Y| = k$. There is an partition $(X_1,X_2)$ of $X$ induced by $\tau$ such that $\frac{k}{3} \leq |X_1 \cap Y| \leq \frac{2k}{3}$ and $\frac{k}{3} \leq |X_2 \cap Y| \leq \frac{2k}{3}$.
\end{lemma}
\begin{proof}
Consider a node $t$ of $\tau$ with two children $t_1$ and $t_2$ such that $|X_t \cap Y| > \frac{2k}{3}$ (so $|\overline{X}_t \cap Y| < \frac{k}{3}$). There must be an $i \in \{1,2\}$ such that $|X_{t_i} \cap Y| > \frac{k}{3}$ for otherwise $|X_t \cap Y| > \frac{2k}{3}$ would not hold. Now there are two cases. In case (1) we have $|X_{t_i} \cap Y| \leq \frac{2k}{3}$ and we are done: the partition $(X_{t_i}, X_{t_{3-i}}\cup \overline{X}_t)$ obtained by removing $\{t_i,t\}$ from $\tau$ matches the statement of the lemma. In case (2) we have $|X_{t_i} \cap Y | > \frac{2k}{3}$, then we repeat the argument below $t_i$. Clearly we have at least one node $t \in T$ such that $|X_t \cap Y| > \frac{2k}{3}$: the root of $T$. Since the vtree has finite depth, at some point we will not be able to go down the vtree so (1) will hold instead of (2).
\end{proof}

We make Charlotte and Adam play the game on $F_{G,k}$. Charlotte chooses a model $\alpha$ and a vtree $\tau$. By Lemma~\ref{lemma:models_are_ordered_k_cliques}, $\alpha = \alpha_S$ for some ordered $k$-clique $S = (s_1,\dots,s_k)$. Let $Y := \{x_i[s_i s_{i + 1}] \mid i \in [k-1]\} \cup \{x_k[s_ks_1]\}$. Adam chooses the cut from Lemma~\ref{lemma:goodCut}. Say, w.l.o.g., that $X_1 \cap Y = \{x_1[s_1 s_2],x_2[s_2s_3],\dots,x_\ell[s_\ell,s_{\ell+1}]\}$ for some $\ell$ between $\frac{k}{3}$ and $\frac{2k}{3}$. Now Charlotte must construct a rectangle $r(X) = \rho_1(X_1) \land \rho_2(X_2)$ accepting only satisfying assignments of $F_{G,k}$ and in particular accepting $\alpha_S$.

\begin{lemma}\label{lemma:models_of_rectanlge}
Let $r = \rho_1 \land \rho_2$ be a rectangle with respect to the partition $(X_1,X_2)$ in a rectangle cover of $F_{G,c}$. Suppose $r$ is satisfied by $\alpha_S$ and $\alpha_U$ for $U = (u_1,\dots,u_k)$ and $S = (s_1,\dots,s_k)$ two ordered $k$-cliques of $G$. If $U[1,\ell] \neq S[1,\ell]$ and $U[\ell+1,k] \neq S[\ell+1,k]$, then $\elements(U) = \elements(S)$.
\end{lemma}
\begin{proof}
Let $i \in [1,\ell]$ such that $u_i \neq s_i$ and let $j \in [\ell+1,k]$ such that $u_j \neq s_j$.  $i$ and $j$ are fixed in this proof. Note that it is in theory possible that $u_i = s_j$ and $s_i = u_j$. Write $\alpha_S = \alpha^1_S \alpha^2_S$ and $\alpha_{U} = \alpha^1_U\alpha^2_U$ to distinguish the parts of the assignments over $X_1$ and the parts over $X_2$. By definition $\rho_1$ accepts $\alpha^1_S$ and $\alpha^1_U$, and $\rho_2$ accepts $\alpha^2_S$ and $\alpha^2_U$. So $r$ accepts $\beta := \alpha^1_S\alpha^2_U$ and $\beta' := \alpha^1_U\alpha^2_S$. $r$ accepts only satisfying assignments of $F_{G,k}$ so we have $\beta = \alpha_W$ and $\beta' = \alpha_{W'}$ for some ordered $k$-cliques $W$ and $W'$ of $G$.  

Suppose $x_j[u_ju_h] \in X_1$ for some $h \in [k]$. Then $\beta'(x_j[u_ju_h]) = \alpha^1_U(x_j[u_ju_h]) = 1$ and thus $W'[j] = u_j$. But, by definition, $x_j[s_j,s_{j+1}]$ is in $X_2$ and $\alpha_S(x_j[s_j,s_{j+1}]) = 1$, so $\beta'(x_j[s_j,s_{j+1}]) = \alpha^2_S(x_j[s_j,s_{j+1}]) = 1$ and thus $W'[j] = s_j$. This is a contradiction since $u_j \neq s_j$. So the following holds
\begin{itemize}
\item[(1)] for every $h \in [k] \setminus \{j\}$, $x_j[u_ju_h]$ is in $X_2$
\end{itemize}
A similar argument shows that 
\begin{itemize}
\item[(2)] for every $h \in [k] \setminus \{i\}$, $x_i[u_iu_h]$ is in $X_1$
\end{itemize}
Now suppose $x_i[s_is_h] \in X_2$ for some $h \in [k]$. Then $\beta'(x_i[s_is_h]) = \alpha^2_S(x_i[s_is_h]) = 1$ and thus $W'[i] = s_i$. But by (2), we also have that $x_i[u_iu_p]$ is in $X_1$ for every $p \in [k]\setminus \{i\}$, so $\beta'(x_i[u_iu_p]) = \alpha^1_U(x_i[u_iu_p]) = 1$ and thus $W'[i] = u_i$. This is a contradiction since $u_i \neq s_i$. So the following holds
\begin{itemize}
\item[(3)] for every $h \in [k] \setminus \{i\}$, $x_i[s_is_h]$ is in $X_1$
\end{itemize}
Now let us recall that 
\begin{itemize}
\item[(4)] if $W$ is an ordered $k$-clique then for every $v,w \in V(G)$ and every $h$, $\alpha_W(x_h[vw]) 1 \Rightarrow v \in \elements(W) \text{ and } w \in \elements(W)$
\end{itemize}
So when we look at $\beta= \alpha^1_S\alpha^2_U = \alpha_W$, given (1) and (4) we obtain that $\elements(W) = \elements(U)$, and given (3) and (4) we obtain that $\elements(W) = \elements(S)$. Thus $\elements(S) = \elements(U)$.
\end{proof}

It follows from Lemma~\ref{lemma:models_of_rectanlge} that the number of assignments accepted by $r$ is at most ``number of ordered $k$-cliques of $G$ of the form $(s_1,\dots,s_\ell,\ast)$'' + ``number of ordered $k$-cliques of $G$ of the form $(\ast,s_{\ell+1},\dots,s_k)$'' + ``number of ordered $k$-cliques of $G$ whose elements are that of $S$''. Hence
\begin{equation}\label{equation:upper_bound_rectangle}
|r^{-1}(1)| \leq \binom{n - \ell}{k - \ell}(k-\ell)! + \binom{n - k + \ell}{\ell}\ell! + k!
\end{equation}

\begin{lemma}\label{lemma:calculus}
Every rectangle $r$ in a rectangle cover of $F_{G,k}$ with respect to a partition given by Lemma~\ref{lemma:goodCut} verifies $\binom{n}{k}k!/|r^{-1}| = (n/k)^{\Omega(k)}$.
\end{lemma}
\begin{proof}
Start from Equation~(\ref{equation:upper_bound_rectangle}).
Let $h := k - \ell$. Both $\ell$ and $h$ are greater than $k/3$ so $\binom{n}{k}k!/|r^{-1}(1)|$ is at least
\begin{multline*}
\frac{n!}{(n-\ell)! + (n-h)! + (n-k)!k!} 
 \geq \frac{n!}{2(n-\frac{k}{3})! + (n-k)!k!} 
\geq \frac{n!}{3\max((n-\frac{k}{3})!,(n-k)!k!)}
\\ = \frac{1}{3}\min\left(\frac{n!}{(n-\frac{k}{3})!}, \binom{n}{k}\right) \geq \frac{1}{3}\min\left(\binom{n}{k/3}, \binom{n}{k}\right) = \frac{1}{3}\binom{n}{k/3} \geq \frac{1}{3}\left(\frac{3n}{k}\right)^{k/3}
\end{multline*}
\end{proof}

When $G$ is the complete graph, $|F_{G,k}^{-1}(1)| = \binom{n}{k}k!$ so, if Adam always chooses cuts using Lemma~\ref{lemma:goodCut}, then the game takes $(n/k)^{\Omega(k)}$ rounds to complete, regardless of Charlotte's choices because, by Lemma~\ref{lemma:calculus}, $(n/k)^{\Omega(k)}$ rectangles are required to cover all satisfying assignments of $F_{G,k}$. We have shown with Lemmas~\ref{lemma:chi_i_are_slim_for_complete_OBDD},~\ref{lemma:chi_uv_are_slim_for_complete_OBDD} and~\ref{lemma:tw_at_most_k} that $F_{G,k}$ verifies the conditions of Theorem~\ref{theorem:lowerBound}. Using Theorem~\ref{theorem:adversarialLowerBound} finishes the proof of Theorem~\ref{theorem:lowerBound}.

\section{Faster Model Counting}\label{section:model_counting}

In this section, we show that model counting for three special cases
of systems of constraints that are $O(1)$-slim for OBDDs can be done
faster, compared to compiling the system as in Section~\ref{section:fpt_compilation} and counting from the compiled form. To this aim, we provide an FPT algorithm based on dynamic programming over a nice tree-decomposition of the incident graph of the system. We start with describing the algorithm for systems of literal-symmetric functions that are $O(1)$-slim for OBDDs, and then show that the speed of the algorithm can be improved when we consider even more particular cases, namely systems of one-sided constraints and systems of a mixture of CNF clauses and modulo constraints. All results in this section are stated in the \emph{unit-cost model}, where every arithmetic operation is counted as an elementary operation.

\subsection{Algorithm}
In Section \ref{section:STS}, we have shown that
literal-symmetric functions that are $O(1)$-slim for OBDDs are
precisely the functions described by a CSTS of constant size modulo literal-flipping. We use the latter representation and provide an FPT algorithm for model counting for a system of such constraints. The algorithm is based on dynamic programming over a nice tree decomposition of the incident graph of the system and is parameterized by the treewidth $k$ of the incident graph of the system and the maximum state size $w$ of its constraints. Throughout the entire section, we assume $w\geq 2$. The main theorem of this subsection is stated here, whose proof we will give later on in this chapter, after sufficient technical details and supporting lemmas have been provided.

\begin{theorem}\label{dynamic-brute-force-time}
Let $F$ be a system of constraints whose maximum state size is $w$, and let $\calT = (T, b)$ be a nice width-$k$ tree decomposition of the incidence graph $G_F$. Then, given $F$ and $\calT$, one can count the number of models of $F$ with \mbox{$O(w^{2k}\cdot |G_F|)$} elementary operations. 
\end{theorem}

The idea for the algorithm is to perform dynamic programming over the nice tree decomposition. 
Let $F$, $\calT$ be as in Theorem \ref{dynamic-brute-force-time}. For each node $t$ of $T$, let $T_t$ denote the subtree of $T$ rooted at $t$, let $F_t$ and $V_t$ be the set constraints and the set of variables that appear in the bags of $T_t$, respectively. We write $b_F(t) = b(t) \cap F$, and $b_V(t) = b(t) \cap \var(F)$ for the set of constraints and variables in $b(t)$, respectively. For all $t\in T$, let $\widetilde{V_t}:=V_t \setminus b_V(t)$. 
For each $c\in b_F(t)$, let $(S^c,f_0^c,f_1^c,s_0^c,T^c)$ be a minimal CSTS that describes $c$ modulo literal-flipping and let $\phi^c$ be the literal-flipping that witnesses this. We fix this choice of CSTS before the start of the algorithm. Let $\delta^c$ be the corresponding extended transition function. For each $F' \subseteq F$, let $\mathcal{S}(F'):=\{R\subseteq\bigcup_{c\in F'}\{c\}\times S^c\,|\,\text{for each }c\in F'\text{ there is a unique }s\text{ s.t. }(c,s)\in R\}$.
Given $F'\subseteq F$, $c \in F'$ and $\bar{s}\in\mathcal{S}(F')$, we use $s^c$ to denote the unique $s$ such that $(c,s)\in\bar{s}$. Define a binary operation $+$ on $\mathcal{S}(F')$ by $\bar{s}_1+\bar{s}_2:=\bar{s}$ where $s^c=s_1^c+s_2^c$. 
For each $c \in F$, we use $c^+$ to denote the set of variables $x_i$ that occur in $c$ with $\phi^c(x_i)=x_i$, and $c^-$ the set of variables $x_i$ that occur in $c$ with $\phi^c(x_i)=\bar{x}_i$. For any partial assignment $\tau$ to the variables, let $q_0(\tau,c):=|\{v\in c^+\,|\,\tau(v)=0\}| + |\{v\in c^-\,|\,\tau(v)=1\}|$ and $q_1(\tau,c):=|\{v\in c^+\,|\,\tau(v)=1\}| + |\{v\in c^-\,|\,\tau(v)=0\}|$. Finally let $\delta^c[\tau]:=\delta^c(1^{q_1(\tau,c)}0^{q_0(\tau,c)})$.

\begin{definition}
For each assignment $\alpha: b_V(t) \rightarrow\{0, 1\}$ and $\bar{s}\in \mathcal{S}(b_F(t))$, we define $N(t, \alpha,\bar{s})$ as the set of assignments $\tau : V_t \rightarrow \{0, 1\}$ for which the following conditions hold:
\begin{enumerate}
\item $\tau(v) = \alpha(v)$ for all variables $v \in b_V(t)$.
\item For each $c\in b_F(t)$, $\delta^c[\tau|_{\widetilde{V_t}}]=s^c$.
\item For each $c\in F_t \setminus b_F(t)$,  $\delta^c[\tau]\in T^c$. 
\end{enumerate}
\label{N-def}
\end{definition}
We represent the values of $n(t, \alpha, \bar{s}) = |N(t, \alpha, \bar{s})|$ for all $\alpha: b_V(t) \rightarrow \{0, 1\}$ and $\bar{s}\in\mathcal{S}(b_F(t))$ by a table $M_t$ with $|b(t)| + 1$ columns and $2^{|b_V(t)|}\cdot|\mathcal{S}(b_F(t))| \leq 2^{|b_V(t)|}\cdot w^{|b_F(t)|} \leq w^k$ rows. The first $|b_V(t)|$ columns of $M_t$ contain Boolean values encoding $\alpha(v)$ for variables $v \in b_V(t)$, followed by $|b_F(t)|$ columns, one for each $c\in b_F(t)$ with the entry $s^c$. The last column contains the integer $n(t,\alpha, \bar{s})$. 

The following lemmas show how to compute the table $M_t$ for a node $t \in T$ depending on its type, assuming the tables of its children have already been computed.

\begin{lemma}\label{join-lemma}
    Let $t$ be a join node with children $t_1,t_2$. For each assignment $\alpha: b_V (t') \rightarrow \{0, 1\}$, and $\bar{s}\in\mathcal{S}(b_F(t))$, we have
    \[n(t,\alpha, \bar{s}) = \sum_{\bar{s}_1 + \bar{s}_2=\bar{s}} n(t_1,\alpha,\bar{s}_1)\ \cdot\ n(t_2,\alpha,\bar{s}_2).\]    
    \end{lemma}
    
    \begin{proof}
    First, note that for each truth assignment $\alpha:b_V (t) \rightarrow \{0, 1\}$ and $i\in\{1,2\}$, if $\bar{s_1}\not=\bar{s_2}$, then $N(t_1,\alpha,\bar{s_1}) \cap N(t_2,\alpha,\bar{s_2}) = \varnothing$. This means that it suffice to show that the mapping
    \[f:N(t,\alpha,\bar{s})\rightarrow\bigcup\limits_{\forall c\in b_C(t).s_1^c\cdot s_2^c=s^c} N(t_1,\alpha,\bar{s_1})\ \times\ N(t_2,\alpha,\bar{s_2}),\]$\tau\mapsto(\tau|_{V_{t_1}},\tau|_{V_{t_2}})$
    is a bijection. The equality above follows then from this claim immediately. 
    
    To prove the claim, we first show that $f$ is well-defined. Let $\tau\in N(t,\alpha,\bar{s})$. Define $\bar{s_i}:=\{(c,s)\,|\,c\in b_C(t), s=\delta^c[\tau|_{V_{t_i}}/b_V(t)]\}$ for $i = 1,2$. To see that for all $c\in b_C(t)$ we have $s_1^c + s_2^c = s^c$, note that since $(T,b,r)$ is a nice tree decomposition, it follows that 
    $\widetilde{V_{t_1}}\cap\widetilde{V_{t_2}}=\varnothing$ and $\widetilde{V_{t_1}}\cup\widetilde{V_{t_2}}=\widetilde{V_{t}}$.
    This means that for any given $c\in b_C(t)$, 
    if $\delta^c[\tau|_{\widetilde{V_{t_1}}}]=\delta^c(1^{a_1}0^{b_1})$ and $\delta^c[\tau|_{\widetilde{V_{t_2}}}]=\delta^c(1^{a_2}0^{b_2})$, 
    then $\delta^c[\tau|_{\widetilde{V_{t}}}]=\delta^c(1^{a_1 + a_2}0^{b_1+b_2}) = \delta^c(1^{a_1}0^{b_1})+\delta^c(1^{a_2}0^{b_2})$. 
    
    Now, we show that $\tau|_{V_{t_1}}$ satisfies the three conditions in Definition \ref{N-def}. The argument for $\tau|_{V_{t_2}}$ is similar. The first two conditions are direct consequences of definition. Let $c\in C_t/b_C(t)$. Then $c$ is either forgotten in $T_{t_1}$ or in $T_{t_2}$. Without loss of generality, assume $c$ is forgotten at $t'\in T_{t_1}$. Since $(T,b,r)$ is a nice tree decomposition, it follows that $\text{var}(c)\subseteq V_{t'}\subseteq V_{t_1}$ and $V_{t_1}\cap\widetilde{V_{t_2}}= \varnothing$. Therefore, $\delta^{c'}[\tau|_{V_{t_1}}]=\delta^{c'}[\tau|_{V_{t'}}]$, $\text{var}(c)\cap\widetilde{V_{t_2}} = \varnothing$, and in turn, $\delta^{c'}[\tau|_{\widetilde{V_{t_2}}}]=\delta^{c'}(\epsilon)$. Now we conclude that $\delta^{c'}[\tau|_{V_{t_1}}]\in T^{c'}$ because 
    \begin{equation*}
    \delta^{c'}[\tau|_{V_{t_1}}]=\delta^{c'}[\tau|_{V_{t_1}}]+ s_0^{c'}=\delta^{c'}[\tau|_{V_{t_1}}]+\delta^{c'}(\epsilon)=\delta^{c'}[\tau|_{V_{t_1}}]+\delta^{c'}[\tau|_{\widetilde{V_{t_2}}}]=\delta^{c'}[\tau]\in T^{C'}.    
    \end{equation*}
    
    So far, we have shown that $f$ is well-defined. To see that $f$ is injective, suppose $\tau,\tau': V_{t}\rightarrow\{0,1\}$, $\tau|_{V_{t_1}} = \tau'|_{V_{t_1}}$ and $\tau|_{V_{t_2}} = \tau'|_{V_{t_2}}$. We distinguish two cases. If $v\in b_V(t_1)$, then $\tau(v)=\tau|_{V_{t_1}}(v) = \tau'|_{V_{t_1}}(v)=\tau'(v)$. Otherwise, $v\in b_V(t_2)$. 
    Then $\tau(v)=\tau|_{V_{t_2}}(v) = \tau'|_{V_{t_2}}(v)=\tau'(v)$. Finally, we show that $f$ is surjective. Let $\tau_1\in N(t_1,\alpha,\bar{s_1})$, $\tau_2\in N(t_2,\alpha,\bar{s_2})$ for some $\bar{s_1},\bar{s_2}\in\mathcal{S}(b_C(t))$ such that for all $c\in b_C(t)$ we have $s_1^c\cdot s^c_2=s^c$. 
    We show that $\tau:=\tau_1\cup\tau_2\in N(t,\alpha,\bar{s})$. The first condition of Definition \ref{N-def} is satisfied because for all $v\in b_V(T)$ we have $\tau(v) = \tau_1(v)=\tau_2(v)=\alpha(v)$. The second condition is satisfied because $\widetilde{V_{t_1}}\cap\widetilde{V_{t_2}}=\varnothing$ and therefore for all $c\in b_C(t)$, 
    \begin{equation*}
    \delta^c[\tau|_{\widetilde{V_{t}}}] = \delta^c[\tau|_{\widetilde{V_{t_1}}}]+\delta^c[\tau|_{\widetilde{V_{t_2}}}]= \delta^c[\tau_1|_{\widetilde{V_{t_1}}}]+\delta^c[\tau_2|_{\widetilde{V_{t_2}}}] = s_1^c+s_2^c = s^c    
    \end{equation*}
    Thus, $f$ is surjective.
\end{proof}

\begin{lemma}
    \label{introduceLemma}
    Let $t$ be an introduce node with child $t'$. For each truth assignment $\alpha: b_V (t') \rightarrow \{0, 1\}$, and $\bar{s}\in\mathcal{S}(b_F(t'))$, we have the following equalities depending on whether it is a variable or a constraint that is introduced at $t$.
    \begin{enumerate}
        \item[\normalfont{1.}] A variable $v$ is introduced. We have
        \[n(t,\alpha\cup\{(v,0)\},\bar{s})=n(t,\alpha\cup\{(v,1)\},\bar{s})=n(t',\alpha,\bar{s}).\]
        
        \item[\normalfont{2.}] A constraint $c$ is introduced. We have
         \[n(t,\alpha,\bar{s}\cup\{(c,s)\})=\begin{cases} n(t',\alpha,\bar{s}),&\text{if }s=s_0^{c};\\
        0, & otherwise.
        \end{cases}\]
    \end{enumerate}
    \end{lemma}
    \begin{proof}
        \begin{enumerate}
            \item It suffices to show that for arbitary $\tau:V_{t'}\rightarrow\{0,1\}$, $\tau':=\tau\cup\{(v,0)\}\in N(t,\alpha\cup\{(v,0)\},\bar{s})$ if and only if $\tau\in N(t,\alpha\cup\{(v,0)\},\bar{s})$. This equivalence follows from the following observations. 
            \begin{enumerate}
            \item $\tau(x)=\alpha(x)$ for all $x\in b_V(t')$ if and only if $\tau'(x)=\alpha(x)\cup\{(v,0)\}$ for all $x\in b_V(t)$.
            \item For each $c\in b_C(t)$, $\tau'|_{\widetilde{V_t}}=\tau|_{\widetilde{V_{t'}}}$ and therefore, $\delta^c[\tau'|_{\widetilde{V_t}}] = \delta^c[\tau|_{\widetilde{V_{t'}}}]$.
            \item Since $(T, b,r)$ is a nice tree decomposition, it follows that for each $c\in C_t/ b_C(t)$, we have $v\not\in\text{var}(c)$. Therefore, $c\in C_t/ b_C(t)$ we have $\delta^c[\tau']=\delta^c[\tau]$.
            \end{enumerate}
            
            \item Since $(T, b,r)$ is a nice tree decomposition, it follows that $\text{var}(c')\cap\widetilde{V_t} = \varnothing$ and therefore $\delta^{c'}[\tau] = s_0^{c'}$ for any $\tau:V_t\rightarrow\{0,1\}$. Now it is easy to check that for any $\tau:V_t\rightarrow\{0,1\}$, $\tau\in N(t,\alpha,\bar{s}\cup\{(c',s_0^{c'})\})$ if and only if $\tau\in N(t',\alpha,\bar{s}\})$.
        \end{enumerate}
    \end{proof}

\begin{lemma}
    \label{forgetLemma}
    Let $t$ be a forget node with child $t'$. For each truth assignment $\alpha:b_V (t) \rightarrow \{0, 1\}$, and $\bar{s}\in\mathcal{S}(b_F(t))$, we have the following two equalities depending on whether it is a variable or a constraint that is forgotten at $t$.
    \begin{enumerate}
        \item[\normalfont{1.}] A variable $v$ is forgotten. Let $\alpha_0$ and $\alpha_1$ denote the extension to $\alpha$ that sets $v$ to $0$ and $1$, respectively. We have
    \end{enumerate}     
    \[ n(t,\alpha,\bar{s})=\sum_{\bar{s}'\in\sigma_0(v,\bar{s})}n(t',\alpha_0,\bar{s}') + \sum_{\bar{s}'\in \sigma_1(v,\bar{s})}n(t',\alpha_1,\bar{s}')\]
    \begin{enumerate} \item[] where $\sigma_1(v,\bar s')$ is the set  $\Biggl\{\bar s'\in\mathcal{S}(b_F(t))\, \Bigg| \,s^c=
        \begin{cases}
        f^c_{0}({s'}^c)&v\in c^-\\
        f^c_1({s'}^c)&v\in c^+\\
        {s'}^c& v \not\in \var(c)
        \end{cases}\Biggl\}$
        and $\sigma_0$ is defined similarly by exchanging $f^c_0$ and $f^c_1$.
        \item[\normalfont{2.}] A constraint $c$ is forgotten. We have
        \[n(t,\alpha,\bar{s})=\sum_{\delta^{c'}[\alpha]+s\,\in T^{c'}}n(t',\alpha,\bar{s}\cup\{(c',s)\}).\]
    \end{enumerate}    
    \end{lemma}
    
    \begin{proof}
    \begin{enumerate}
        \item It suffices to show that for any $\tau:V_{t'}\rightarrow\{0,1\}$ such that $\tau(v)=0$, we have $\tau\in N(t,\alpha,\bar{s})$ if and only if $\tau\in N(t',\alpha\cup\{(v,0)\},\bar{s'})$ for some $\bar{s'}\in F^{-1}_0(v,\bar{s})$. This equivalence follows from the following observation. For each $c\in b_C(t)$, $\delta^c[\tau|_{\widetilde{V_t}}]=\delta^c[\tau|_{\widetilde{V_{t'}}}]+\delta^c[(v,\tau(v))]$.
        \item It suffices to show that for any $\tau:V_t\rightarrow\{0,1\}$, $\tau\in N(t,\alpha,\bar{s})$ if and only if $\tau\in N(t',\alpha,\bar{s}\cup\{c',s\})$ for some $s$ such that $\delta^{c'}[\alpha]+s\in T^{c'}$. This equivalence follows from the following observation.
        \begin{enumerate}
        \item For each $c\in b_C(t)$, since $\widetilde{V_t}=\widetilde{V_{t'}}$, it follows that $\delta^c[\tau|_{\widetilde{V_t}}] = \delta^c[\tau|_{\widetilde{V_{t'}}}]$.
        \item $\delta^{c'}[\tau]=\delta^{c'}[\tau|_{ b_V(t)}]+\delta^{c'}[\tau|_{\widetilde{V_t}}] = \delta^{c'}[\alpha]+\delta^{c'}[\tau|_{\widetilde{V_{t'}}}] = \delta^{c'}[\alpha]+s$ 
        \end{enumerate}
    \end{enumerate}
    \end{proof}

    \begin{lemma}\label{leaf-lemma}
    Let $t$ be a leaf node. For each truth assignment $\alpha: b_V(t)\rightarrow\{0, 1\}$ and $\bar{s}\in\mathcal{S}(b_F(t))$, we have
    \[n(t,\alpha,\bar{s}) =\begin{cases}
    1, &\text{if }s^c=s_0^c\text{ for all }c\in b_F(t);\\
    0, &otherwise.\end{cases}\]    
    \end{lemma}
    \begin{proof}
    Since $t$ is a leaf node, it follows that the only $\tau:V_t\rightarrow\{0,1\}$ such that $\tau(v)=\alpha(v)$ for all $v\in b_V(t)$ is $\alpha$. Since $\widetilde{V_t}=\varnothing$, it follows that for each $c\in b_C(t)$ we have $\delta^c[\alpha|\widetilde{V_t}]=\delta^c(\epsilon)=s_0^c$. Because $C_t/ b_C(t)=\varnothing$, the third condition is vacuously true for $\alpha$.
\end{proof}

Now we prove Theorem \ref{dynamic-brute-force-time}.

\begin{proof}
    It is easy to see that $M_t$ for a leaf node $t$ takes as many elementary operations linear to the number of rows in $M_t$, which is at most $w^k$. Calculating $M_t$ for an introduce node or a forget node $t$ also takes at most $O(w^k)$ many elementary operations, because either each $n(t,\alpha,\bar{s})$ appears in the sum of at most one column, or it appears in the sum of at most two columns, but there are at most $w^{k-1}$ rows in $M_{t'}$ in total. Calculating $M_t$ for a join node $t$ takes at most $O(w^{2k})$ elementary operations. An obvious way to do this is, for all $\alpha$, initialize all $n(t,\alpha,\bar{s})$ as $0$, enumerate the pairs $(\bar{s_1},\bar{s_2})$, and add $n(t_1,\alpha,\bar{s_1})\cdot n(t_2,\alpha,\bar{s_2})$ to the corresponding row in $M_t$.
    
    We calculate $M_t$ for each $t\in T$ following the tree structure, and in the end we read off the final result from table $M_r$ (which is a table consisting of $|F^{-1}(1)|$ as the only element by definition). It is known that one can transform efficiently any tree decomposition of width k of a graph with n vertices into a nice tree decomposition of width at most k and at most 4n nodes \cite[Lemma 13.1.3]{kloks1994treewidth}. Thus, the entire algorithm uses $O(w^{2k}\cdot|G_F|)$ elementary operations
    \end{proof}

\subsection{One-Sided and Modulo Constraints}
    A close inspection of the proof of Theorem
    \ref{dynamic-brute-force-time} reveals that the bottleneck lies in the
    computation of the table $M_t$ for a join node $t$. In fact, if we
    restrict ourselves to more particular cases of systems of constraints,
    namely systems of one-sided constraints and systems of disjunctive
    clauses and modulo constraints, we can speed-up the algorithm by  computing the tables for join
    nodes faster, using advanced techniques like the Convolution Theorem and fast Fourier transform
    (see, e.g., \cite{BjorklundHKK07}, \cite{Slivovsky20}). The two main theorems of this subsection are stated here, whose proofs we will give later on in this chapter, after sufficient technical details and supporting lemmas have been provided.

    \begin{theorem}
    \label{theorem:one-sided}
    Let $F$ be a system of one-sided constraints of maximal state size $w$. Given $F$ and a width-$k$ tree decomposition of $G_F$, one can compute the number of models of $F$ in $O((2w)^kk\,\log(w)\cdot|G_F|)$ elementary operations. 
    \end{theorem}

    \begin{theorem}
    \label{theorem:modulo}
    Let $w$ be a natural number. Let $F$ be a system of constraints comprising only clauses and $m$-modulo constraints for possibly different $m\leq w$.
    Given $F$ and a width-$k$ tree decomposition of $G_F$, one can compute the model count of $F$ in $O(w^kk\log(w)\cdot|G_F|)$ elementary operations.    
    \end{theorem}
    
    We call a constraint \emph{one-sided} when there exists an STS that describes it where either $f_0(s)=s$ for every state $s$, or $f_1(s)=s$ for every state $s$. In the first case, we talk of \emph{$1$-only} constraints and STS, and in the second case, we talk of \emph{$0$-only} constraints and STS.  Throughout this section, we assume that for each one-sided constraint $c\in F$, a minimal one-sided STS that describes it is known before the algorithm starts.
    Given a one-sided constraint $c$, if $c$ is a $1$-only (resp. $0$-only) constraint, then let $c(i)$ denote the state the STS is in after receiving $i$ many $1$'s (resp. $i$ many $0$'s).
    \begin{theorem}
    Let $A$ and $B$ be two $d$-dimensional tensors with elements from a ring $R$ of dimensions $N^M_1\times N^M_2\times...\times N^M_d$, $M=A,B$. Let $N_i:=N^A_i+n^B_i$ for all $1\leq i\leq d$ and  $N:=\prod\limits_{1\leq i\leq d}N_i$. Define $\mathbf{n}^A:=(n_1^A,n_2^A,...,n_d^A)$, $\mathbf{n}^B:=(n_1^B,n_2^B,...,n_d^B)$ and $\mathbf{n}:=(n_1,n_2,...,n_d)$ as $d$-dimensional verctors of indices from $0$ to $\mathbf{N}-1$, which is in turn defined as $\mathbf{N}-1=(N_1-1,N_2-1,...,N_d-1)$. Define an $N_1\times N_2\times...\times N_d$ tensor $A*B$ by 
    \[(A*B)_{\mathbf{n}}:=\sum\limits_{\mathbf{n}^A+\mathbf{n}^B=\mathbf{n}}A_{\mathbf{n}^A}\cdot B_{\mathbf{n}^B}\]
    Then $A*B$ can be calculated using $\tilde{O}(N)$\footnote{We use $f = \tilde{O}(g)$ as shorthand for $f = O(g\,log\,g)$.} elementary operations in $R$ by computing $\calF^{-1}_d(\calF_d(\bar{A})\circ\calF_d(\bar{B}))$. Here $\bar{A}$ and $\bar{B}$ are $N_1\times N_2\times...\times N_d$ tensors resulting from padding $A$ and $B$ with zeros. $\calF_d$ and $\calF^{-1}_d$ are the $d$-dimensional Fourier transform and its inverse and $\circ$ is the Hadamard product.
    \label{polynomial-multiplication}
    \end{theorem}
    
\begin{proof}
    Define $\mathbf{k}:=(k_1,k_2,...,k_d)$ as $d$-dimensional a verctor of indices from $0$ to $\mathbf{N}-1$. Let $x$ be a $d$-dimensional tensor. Then $\calF_d(x)$ is defined as
    \[
        \calF_d(x)_\mathbf{k}:=\sum^{N_1-1}_{n_1=0}e^{-\frac{i2\pi k_1 n_1}{N_1}}\sum^{N_2-1}_{n_2=0}e^{-\frac{i2\pi k_2 n_2}{N_2}}\\\cdot\cdot\cdot\sum^{N_d-1}_{n_d=0}e^{-\frac{i2\pi k_d n_d}{N_d}} x_{\mathbf{n}}.
    \]
    The definitions above can be more compactly expressed in vector notation:
    \[\calF_d(x)_\mathbf{k}:=\sum^{\mathbf{N}-1}_{\mathbf{n}=\mathbf{0}}e^{-i2\pi\mathbf{k}\cdot(\mathbf{n}/\mathbf{N})}x_{\mathbf{n}}.\]
    The inverse $\calF_d^{-1}$ is defined by:
    \[\calF_d^{-1}(x)_\mathbf{k}:=\frac{1}{\prod^d_{i=1}N_i}\sum^{\mathbf{N}-1}_{\mathbf{n}=\mathbf{0}}e^{i2\pi\mathbf{k}\cdot(\mathbf{n}/\mathbf{N})}x_{\mathbf{n}}.\]
    Readers can verify that $\calF_d$ and $\calF^{-1}_d$ as defined above are indeed each other's inverse. The following equation proves the convolution theorem for $d$-dimensional Fourier transform:
    \[\begin{aligned}
    \calF_d(A*B)_{\mathbf{k}}&=\sum^{\mathbf{N}-1}_{\mathbf{n}=\mathbf{0}}e^{-i2\pi\mathbf{k}\cdot(\mathbf{n}/\mathbf{N})}(A*B)_{\mathbf{n}}
    =\sum^{\mathbf{N}-1}_{\mathbf{n}=\mathbf{0}}e^{-i2\pi\mathbf{k}\cdot(\mathbf{n}/\mathbf{N})}\sum^{\mathbf{N}-1}_{\mathbf{m}=\mathbf{0}}A_{\mathbf{m}}\cdot B_{\mathbf{n}-\mathbf{m}}\\
    &=\sum^{\mathbf{N}-1}_{\mathbf{m}=\mathbf{0}}A_{\mathbf{m}}\sum^{\mathbf{N}-1}_{\mathbf{n}=\mathbf{0}}e^{-i2\pi\mathbf{k}\cdot(\mathbf{n}/\mathbf{N})}B_{\mathbf{n}-\mathbf{m}}
    =\sum^{\mathbf{N}-1}_{\mathbf{m}=\mathbf{0}}A_{\mathbf{m}}\cdot e^{-i2\pi\mathbf{k}\cdot(\mathbf{m}/\mathbf{N})}\cdot\calF_d(B)_{\mathbf{k}}\\
    &=\calF_d(A)_{\mathbf{k}}\cdot\calF_d(B)_{\mathbf{k}}.
    \end{aligned}\]
    Here the tensor addition and subtraction are point-wise. The $d$-dimensional Fourier transform and its inverse can be computed by the composition of a sequence of one-dimensional Fourier transforms along each dimension, and thus can be computed in time $\tilde{O}(N)$.
    \end{proof}
    Now we prove Theorem \ref{theorem:one-sided}.
    \begin{proof}
    The theorem follows from the proof of Theorem \ref{dynamic-brute-force-time} if we can reduce the number of elementary operations in calculating $M_t$ for a join node $t$ to using only $\tilde{O}((2w)^k)$ elementary operations.  Let $\chi_c(t)=\{c_1,c_2,...,c_d\}$. Given an $\alpha:\chi_V(t)\rightarrow\{0,1\}$, define tensors $A,B$ with dimensions $\prod_{1\leq i\leq d}n_{c_i}$ as follows. $A[m_1,m_2,...,m_d]:=n(t_1,\alpha,\{(c_i,s)\,|\,1\leq i\leq d,s=c_i(m_i)\})$ for all $0\leq m_i<n_{c_i}$,$1\leq i\leq d$. $B$ is defined similarly by replacing $t_1$ with $t_2$. Note that for any $\bar{s}\in\mathcal{S}(\chi_C(t))$, we have $n(t,\alpha,\bar{s})=\sum_{c_i(m_i)=s^{c_i}}A*B[m_1,m_2,...,m_d]$. For each $\alpha$, we first compute the relevant $A*B$. By Theorem \ref{polynomial-multiplication}, can be done with $\tilde{O}(\prod_{1\leq i\leq d}2n_{c_i})$ elementary operations. Then, we compute $n(t,\alpha,\bar{s})$ as the sum defined above, which uses $O(\prod_{1\leq i\leq d}2n_{c_i})$ elementary operations since each element in $A*B$ appears in the sum in exactly one row of $M_t$. Since there are only $2^{|\chi_V(t)|}$ different $\alpha$'s, we can compute $M_t$ for a join node $t$ using $2^{|\chi_V(t)|}\cdot\tilde{O}(\prod_{1\leq i\leq d}2n_{c_i}) = \tilde{O}((2w)^k)$ elementary operations.
    \end{proof}

    The following definitions and theorems are needed before we can prove Theorem \ref{theorem:modulo}.

    \begin{definition}
    Given a set $V$ and a function $f,g:2^V \rightarrow R$, the \emph{zeta transform} $\zeta f$ of $f$ is defined as $(\zeta f)(X) = \sum_{Y\subseteq X} f(Y)$, the \emph{Moebius transform} $\mu f$ of $f$ is defined as $\mu f(x)=\sum_{Y\subseteq X} (-1)^{|Y\cap X|} f(Y)$. Define $f\cup g:2^V\rightarrow R$ by $f\cup g(X):=\sum_{A\cup B=X}f(A)\cdot g(B)$.
    \end{definition}
    
    \begin{theorem}\cite{Slivovsky20}
    Given a set $V$ and a function $f,g:2^V \rightarrow R$, $f\cup g$ can be computed using $\tilde{O}(2^{|V|})$ elementary operations in $R$ with $f\cup g=\mu(\zeta f\cdot \zeta g)$.
    \label{union-product}
    \end{theorem}
    
    \begin{definition}
    Let $G$ be an abelian group, written additively, and $f,g:G\rightarrow R$ be functions. Define $f*g:G\rightarrow R$ by $f*g(x):=\sum_{a+ b=x}f(a)\cdot g(b)$. 
    \end{definition}
    
    \begin{theorem}\cite{oberst2007fast}
    Let $G$ be an abelian group isomorphic to $\prod_{1\leq i\leq r}\mathbb{Z}/\mathbb{Z}d_i$ for some $d_1,...,d_r\in\mathbb{N}$, and $f,g:G\rightarrow R$ be functions. Then $f*g$ can be computed using at most $|G|\cdot(\sum_{1\leq i\leq r}d_i-1)$ many elementary operations in $R$ using the equation $f*g=\calF^{-1}(\calF(f)\cdot\calF(g))$ where $\calF$ and $\calF^{-1}$ are the Fourier transform on finite abelian groups and its inverse.
    \label{convolution-theorem}
    \end{theorem}
    
    \begin{lemma}
        Let $D=\{c_1,c_2,...,c_d\}$ be a system of constraints where $c_i$ is a $n_i$-modulo constraint. Then $(\mathcal{S}(D),+)$ is isomorphic to $\prod_{1\leq i\leq d}\mathbb{Z}/\mathbb{Z}n_i$.
        \label{group-lemma}
    \end{lemma}

    \begin{proof}
    It is easy to see that for any $m$-modulo constraint $c$, $(S^c,+)$ is isomorphic to $\mathbb{Z}/\mathbb{Z}m$. By definition, $(\mathcal{S}(\chi_C(t)),+)$ is the direct product of $(S^c,+)$, $c\in\chi_C(t)$.
    \end{proof}

    Let $m \in \mathbb{N}$, an $m$-modulo constraint $c$ is a constraint such that, for all $x_1,\dots,x_n,x_1',\dots,x_n'\in \{0,1\}$ verifying $\sum_{i\in[n]}x_i = \sum_{i\in[n]}x'_i \mod m$, we have $c(x_1,\dots,x_n)=c(x'_1,\dots,x'_n)$. Note that every $m$-modulo constraint is $1$-only. We assume that the CSTS chosen for a disjunctive clause is always as shown in Figure \ref{STS-example}, where $c(0)=s_0$ and $c(i)=s_1$ for all $i>0$. 

    Now we prove Theorem \ref{theorem:modulo}.

    \begin{proof}
    The theorem follows from the proof of Theorem \ref{dynamic-brute-force-time} if we can reduce the number of elementary operations in calculating $M_t$ for a join node $t$ to using only $\tilde{O}(w^k)$ elementary operations. Let $t_1,t_2$ be the children of $t$.  Fix a partial assignment $\alpha:\chi_V(t)\rightarrow \{0,1\}$. Let $C':=\chi_C(t)\cap C$ and $D':=\chi_C(t)\cap D$. Define $f:2^{C'}\times \mathcal{S}(D')\rightarrow\mathbb{N}$ by $f(X,\bar{s}):=n(t_1,\alpha,\bar{s}\cup\{(c,s)\,|\,c\in C', s=s_1\text{ if }c\in X \text{ and }s=s_0\text{ otherwise}\})$. $g$ is defined similarly by replacing $t_1$ with $t_2$. Let $A:=\{c\in C\,|\,s^c=s_1\}$. We have the following equalities. We use $A_1$ as shorthand for $\{c\in C\,|\,x^c=s_1\}$ and $A_2$ for $\{c\in C\,|\,y^c=s_1\}$. 
    \begin{equation}
    \begin{aligned}
    n(t,\alpha,\bar{s})&=\sum_{\bar{x}+\bar{y}=\bar{s}} n(t_1,\alpha,\bar{x})\cdot n(t_2,\alpha,\bar{y})\\
    &=\sum_{\substack{A_1\cup A_2=A,\\\bar{x}_D+\bar{y}_D=\bar{s}_D}} f(A_1,\bar{x}_D)\cdot g(A_2,\bar{y}_D)\\
    &=\sum_{A_1\cup A_2=A}\bigl(\lambda x.f(A_1, x)*\lambda x.g(A_2, x)\bigl)(\bar{s}_D)\\
    &=\sum_{A_1\cup A_2=A}\calF^{-1}\bigl(\calF(\lambda x.f(A_1, x))\cdot\calF(\lambda x.g(A_2, x))\bigl)(\bar{s}_D)\\
    &=\calF^{-1}\bigl(\lambda y.\calF(\lambda x.f(y, x))\cup\lambda y.\calF(\lambda x.g(y, x))(A)\bigl)(\bar{s}_D)\\
    \end{aligned}
    \label{join_node}
    \end{equation}
    According to \ref{join_node}, we can compute the values of $n(t,\alpha,\bar{s})$ for all $\bar{s}\in\mathcal{S}(C\cup D)$ as follows. First, compute the Fourier transform of $f$ and $g$ point-wise with respect to the first input. By Theorem \ref{convolution-theorem} and Lemma \ref{group-lemma}, this costs $2^{|C'|}\cdot \tilde{O}(w^{|D'|})$ elementary operations. Then, compute the union product the two functions obtained in the previous step point-wise with respect to the second input. By Theorem \ref{union-product}, this costs $w^{|D'|}\cdot \tilde{O}(2^{|C'|})$ elementary operations. Now, compute the inverse Fourier transform of the function obtained in the previous step point-wise with respect to the first input. By Theorem \ref{convolution-theorem} and Lemma \ref{group-lemma}, this costs $2^{|C'|}\cdot \tilde{O}(w^{|D'|})$ elementary operations. Since $|\chi_V(t)| + |C'| + |D'|\leq k$, we conclude that with this process we can compute $M_t$ using $\tilde{O}(w^k)$ elementary operations.
    \end{proof}

    Since XOR constraints are $2$-modulo constraints, it is easy to see the following corollary.
    
    \begin{corollary}
    Let $F$ be a system of clauses and XOR constraints.
    Given $F$ and a width-$k$ tree decomposition of $G_F$, one can compute the model count of $F$ in $O(2^kk\cdot|G_F|)$ elementary operations.    
    \end{corollary}
    
\section{Conclusion}

We have shown that the compilation of systems of constraints
parameterized by incidence treewidth to d-SDNNF is FPT for specific
families of constraints, namely, constraints whose OBDD- and
SDNNF-width are bounded by a constant for all variable orders and all
vtrees.
This generalizes known results for CNF, i.e., systems of disjunctive
clauses, to many more constraints, including modulo and
small-threshold constraints.  Since compilation to d-SDNNF is often
used in practice as a first step towards model counting, we have also
shown that faster FPT model counting algorithms exist without
compilation when we restrict the constraints considered. A natural
question here is whether one can push our results further, that is, to
constraints that do not belong to the families considered in this
paper.
It seems that positive compilation results can always be established by reduction to the compilation of CNF formulas (in this paper, CNF encodings of the constraints). We also ask if there are situations where encoding the problem to CNF before compiling is provably a worse strategy than reasoning on the original problem.    

\section*{Acknowledgments}
The authors would like to thank Stefan Mengel for its insights on proving Theorem~\ref{theorem:lowerBound}.
The research leading to this publication has received funding from the European Union's Horizon 2020 research and innovation programme under grant agreement No.~101034440, and was supported by the Vienna Science and Technology Fund (WWTF) within the project ICT19-065 and from the Austrian Science Fund (FWF) within the projects \fwf{ESP235}, \fwf{P36688}, and \fwf{P36420}. 

\begin{figure}[h!]
\centering
\includegraphics[scale=0.35]{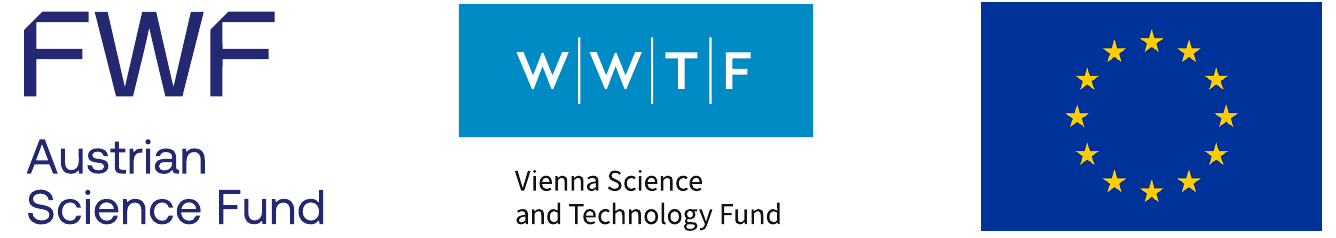}
\end{figure}

\bibliographystyle{alpha}
\bibliography{main}
\end{document}